\documentclass{lmcs} %%% last changed 2014-08-20

\newif\ifdraft
\draftfalse

\usepackage{latexsym,amsmath,amssymb}
\usepackage{mathrsfs} % for \mathscr
\usepackage{url}
\usepackage{xspace}
\usepackage{comment}
\theoremstyle{plain}

\newtheorem{theorem}{Theorem}[section]

\newtheorem{proposition}[theorem]{Proposition}

\newtheorem{claim}{Claim}

\theoremstyle{definition}

\newcommand{\If}{\lq\lq$\Leftarrow$\rq\rq\ \ }           % 1st direction of iff
\newcommand{\OnlyIf}{\lq\lq$\Rightarrow$\rq\rq\ \ }      % 2nd direction of iff

\newcommand{\myi}{(\emph{i})\xspace}
\newcommand{\myii}{(\emph{ii})\xspace}
\newcommand{\myiii}{(\emph{iii})\xspace}

 \newcommand{\B}{\mathcal{B}}
 
\newcommand{\E}{\mathcal{E}} \newcommand{\F}{\mathcal{F}}

 \renewcommand{\L}{\mathcal{L}}

 \newcommand{\T}{\mathcal{T}}
 \newcommand{\V}{\mathcal{V}}
 \newcommand{\X}{\mathcal{X}}

\newcommand{\Tau}{\T}
\newcommand{\Eta}{\E}

\newcommand{\ra}{\rightarrow}

\newcommand{\lora}{\longrightarrow}

\newcommand{\EXPTIME}{\textsc{ExpTime}}
\newcommand{\lang}{\mathscr{L}} % curly calligraphic

\newcommand{\limp}{\rightarrow}
\newcommand{\incl}{\subseteq}

\newcommand{\cert}[1][Q,\V]{\mathit{cert}_{#1}}

\newcommand{\Nat}{{\mathbb N}}
\newcommand{\conc}{\mathord{\cdot}}

\newcommand{\true}{\textbf{true}\xspace}
\newcommand{\false}{\textbf{false}\xspace}

\newcommand{\eword}{\varepsilon}
\newcommand{\At}{\mathbf{A}}
\newcommand{\Atwf}{\mathbf{A}^\mathit{wf}}
\newcommand{\XPath}{\textit{XPath}\xspace}
\newcommand{\RXPath}{\textit{RXPath}\xspace}
\newcommand{\CXPath}{\textit{CXPath}\xspace}
\newcommand{\CoreXPath}{\textit{CoreXPath}\xspace}
\newcommand{\muXPath}{\ensuremath{\mu}\textit{XPath}\xspace}

\newcommand{\DIAM}[2]{\langle #1 \rangle #2}
\newcommand{\BOX}[2]{[#1]#2}
\newcommand{\TEST}[1]{#1?}

\newcommand{\TRUE}{\ensuremath{\mathsf{true}}\xspace}
\newcommand{\FALSE}{\ensuremath{\mathsf{false}}\xspace}
\newcommand{\uu}{\ensuremath{\mathbf{u}}\xspace}

\newcommand{\lfp}{\mathsf{lfp}}
\newcommand{\gfp}{\mathsf{gfp}}

\newcommand{\LFP}[1]{\lfp{\{#1\}}}
\newcommand{\GFP}[1]{\gfp{\{#1\}}}

\newcommand{\Intval}[3]{#3_{#2}^{#1}}
\newcommand{\INTVAL}[3]{(#3)_{#2}^{#1}}

\newcommand{\SigmaA}{\Sigma_a}

\newcommand{\Ts}{T_s}

\newcommand{\query}{q}

\newcommand{\Nominal}[1]{\mathit{N_{#1}}}

\newcommand{\achild}{\ensuremath{\mathsf{child}}\xspace}
\newcommand{\apar}{\ensuremath{\mathsf{parent}}\xspace}
\newcommand{\aright}{\ensuremath{\mathsf{right}}\xspace}
\newcommand{\aleft}{\ensuremath{\mathsf{left}}\xspace}
\newcommand{\afchild}{\ensuremath{\mathsf{fchild}}\xspace}

\newcommand{\ifc}{\mathit{ifc}}
\newcommand{\hfc}{\mathit{hfc}}
\newcommand{\irs}{\mathit{irs}}
\newcommand{\hrs}{\mathit{hrs}}

\newcommand{\ini}{s_\mathit{ini}}
\newcommand{\str}{s_\mathit{struc}}

\newcommand{\CL}{\mathit{CL}}
\newcommand{\nnf}{\mathit{nnf}}

\newcommand{\Dom}[1][T]{\Delta^{#1}}            % D "power" T,
\newcommand{\Int}[2][T]{#2^{#1}}                % argument^T
\newcommand{\INT}[2][T]{(#2)^{#1}}              % (argument)^T,

\newcommand{\lab}[1][T]{\ell^{#1}}
\newcommand{\Lab}[2][T]{\lab[#1](#2)}
\newcommand{\Tree}[1][T]{(\Dom[#1],\lab[#1])}

\newcommand{\sta}{\mathit{state}}

\newcommand{\ared}{\ensuremath{\mathtt{red}}\xspace}
\newcommand{\ablue}{\ensuremath{\mathtt{blue}}\xspace}


\begin{document}

\title{Fixpoint Node Selection Query Languages for Trees}

\author[D.~Calvanese]{Diego Calvanese}
%\orcid{0000-0001-5174-9693}
\address{Free University of Bozen-Bolzano\\
 Faculty of Computer Science\\
 Piazza Domenicani 3\\
 39100 Bolzano, Italy
}
\email{calvanese@inf.unibz.it}

\author[G.~De Giacomo]{Giuseppe De Giacomo}
%\orcid{0000-0001-9680-7658}
\address{Sapienza Universit\`a di Roma\\
 Dipartimento di Ingegneria Informatica, Automatica e Gestionale\\
 Via Ariosto 25\\
 00185 Rome, Italy
}
\email{degiacomo@diag.uniroma1.it}

\author[M.~Lenzerini]{Maurizio Lenzerini}
%\orcid{0000-0003-2875-6187}
\address{Sapienza Universit\`a di Roma\\
 Dipartimento di Ingegneria Informatica, Automatica e Gestionale\\
 Via Ariosto 25\\
 00185 Rome, Italy
}
\email{lenzerioni@diag.uniroma1.it}

\author[M.Y.~Vardi]{Moshe Y. Vardi}
%\orcid{0000-0002-0661-5773}
\address{Rice University\\
 Department of Computer Science\\
 S.\ Main Street\\
 77005-1892 Houston, TX, U.S.A.
}
\email{vardi@cs.rice.edu}

%%%%%%%%%%%%%%%%%%%%%%%%%%%%%%%%%%%%%%%%%%%%%%%%%%%%%%%%%%%%%%%%%%%%%%%%%%%%%%
%%% Time-stamp: "2015-02-19 18:49:30 calvanese"
%%%%%%%%%%%%%%%%%%%%%%%%%%%%%%%%%%%%%%%%%%%%%%%%%%%%%%%%%%%%%%%%%%%%%%%%%%%%%%

\begin{abstract}
  The study of node selection query languages for (finite) trees has been a
  major topic in the research on query languages for Web documents. On
  one hand, there has been an extensive study of XPath and its various
  extensions. On the other hand, query languages based on classical logics,
  such as first-order logic (FO) or Monadic Second-Order Logic (MSO), have been
  considered. Results in this area typically relate an XPath-based language to
  a classical logic. What has yet to emerge is an XPath-related language that
  is as expressive as MSO, and at the same time enjoys the computational
  properties of XPath, which are linear time query evaluation and exponential
  time query-containment test. In this paper we propose \muXPath, which is the
  alternation-free fragment of XPath extended with fixpoint operators. Using
  two-way alternating automata, we show that this language does combine desired
  expressiveness and computational properties, placing it as an attractive
  candidate for the definite node-selection query language for trees.
\end{abstract}

\maketitle

\section{Introduction}
\label{sec:introduction}

%MYV: Added JSON
XML\footnote{\url{http://www.w3.org/TR/REC-xml/}\label{footnote-xml}}
and JSON\footnote{\url{http://www.json.org/}\label{footnote-json}}
are standard languages for Web documents supporting semistructured
data.  From the conceptual point of view, an XML or JSON document can be seen
as a finite node-labeled tree, and several formalisms have been proposed
as query languages over documents considered as finite trees.

Broadly speaking, there are two main classes of such languages, those focusing
on selecting a set of nodes based on structural properties of the
tree~\cite{NeSc02,Vian01,CMPP15}, and those where the mechanisms for the
selection of the result also take into account node attributes with their
associated values taken from a specified domain, and the relationship between
attribute values for different
nodes~\cite{BMSS09,BDMSS11,Bouy02,DeLa09,NeSV04,KaTa08}.  We focus here on the
former class of queries, which we call \emph{node selection queries}. Many of
such formalisms come from the tradition of modal logics, similarly to the most
expressive languages of the Description Logics family~\cite{BCMNP03}, based on
the correspondence between the tree edges and the accessibility relation used
in the interpretation strcutures of modal
logics. \XPath~\cite{W3Crec-XPath-1.0} is a notable example of such formalisms,
and, in this sense, it can also be seen as an expressive Description Logic over
finite trees. Relevant extensions of \XPath are inspired by the family of
Propositional Dynamic Logic (PDL)~\cite{Prat78}. For example, \RXPath is the
extension of \XPath with binary relations specified through regular expression,
used to formulate expressive navigational patterns over XML
documents~\cite{CDLV09}.
%MYV: Deleted
%Here, the correspondence is between programs of PDL and paths in the tree.

A main line of research on node selection queries has been on identifying nice
computational properties of \XPath, and studying extensions of such language
that still enjoy these properties. An important feature of \XPath is the
tractability of query evaluation in data complexity, i.e., with respect to the
size of the input tree.  In fact, queries in the navigational core \CoreXPath
can be evaluated in time that is linear in the \emph{product} of the size of
the query and the size of the input tree~\cite{GoKP05,BoPa08}. This property is
enjoyed also by various extensions of \XPath, including
\RXPath~\cite{Marx04b}. Another nice computational property of \XPath is that
checking query containment, which is the basic task for static analysis of
queries, is in \EXPTIME~\cite{NeSc03,Schw04}.  This property holds also for
\RXPath~\cite{TeSe08,CDLV09}, and other extensions of \XPath~\cite{TeLu09}.

Another line of research focused on expressive power. Marx has shown that
\XPath is expressively equivalent to FO$^2$, the 2-variable fragment of
first-order logic, while \CXPath, which is the extension of \XPath with
conditional axis relations, is expressively equivalent to full
FO~\cite{Marx04b,Marx05}. Regular extensions of \XPath are expressively
equivalent to extensions of FO with transitive
closure~\cite{Tenc06,TeSe08}. Another classical logic is Monadic Second-Order
Logic (MSO).  This logic is more expressive than FO and its extensions by
transitive closure~\cite{Libk06b,Tenc06,TeSe08}.  In fact, it has been argued
that MSO has the right expressiveness required for Web information extraction
and hence can serve as a yardstick for evaluating and comparing
wrappers~\cite{GoKo04}.  Various logics are known to have the same expressive
power as MSO, cf.~\cite{Libk06b}, but so far no natural extension of \XPath
that is expressively equivalent to MSO and enjoys the nice computational
properties of \XPath has been identified.

A further line of research focuses on the relationship between query languages
for finite trees and tree automata~\cite{LiSi08,Neve02,Schw07}.  Various
automata models have been proposed. Among the cleanest models is that of
node-selecting tree automata, which are automata on finite trees, augmented
with node selecting states~\cite{NeSc02,FrGK03}.  What has been missing in this
line of inquiry is a single automaton model that can be used \emph{both} for
testing query containment and for query evaluation~\cite{Schw07}. For example,
node-selecting tree automata are adequate for query evaluation, but not for,
say, query containment~\cite{Schw07}.

%MYV: Full version of two papers!
The goal of this paper\footnote{This paper is based on two earlier conference
 papers~\cite{CDLV09,CDLV10}, which did not contain constructions, algorithms,
 and proofs in full detail.}  is dual. We first introduce a comprehensive
automata-theoretic framework for both query evaluation and reasoning about
queries.  The framework is based on \emph{two-way weak alternating tree
 automata}, denoted 2WATAs~\cite{KuVW00}, but specialized for finite trees,
while allowing for infinite runs.  Our main result here is that 2WATAs enable
linear-time query evaluation and exponential time query containment.  This
solves the problem mentioned above of identifying an automaton model that can
be used both for testing query containment and for query
evaluation~\cite{Schw07}.

We then introduce a declarative query language, namely \muXPath, based on
\XPath enriched with \emph{alternation-free} fixpoint operators, which
preserves these nice computational properties.  The significance of this
extension is due to a further key result of this paper, which shows that on
\emph{finite} trees alternation-free fixpoint operators are sufficient to
capture all of MSO, which is considered to be the benchmark query language on
tree-structured data \cite{GoKo04}.  Alternation freedom implies that the least
and greatest fixpoint operators do not interact, and is known to yield
computationally amenable logics~\cite{BCMDH92,KuVW00}.  It is also known that
unfettered interaction between least and greatest fixpoint operators results in
formulas that are very difficult for people to comprehend, cf.~\cite{Koze83}.

Fixpoint operators have been studied in the $\mu$-calculus, interpreted over
arbitrary structures~\cite{Koze83}, which by the tree-model property of this
logic, can be restricted to be interpreted over infinite trees,
%MYV: added
where it is known that the $\mu$-calculus is equivalent to MSO,
cf.~\cite{EmJu91,Niwi97}.  It is also known that, to obtain the full expressive
power of MSO on infinite trees, arbitrary alternations of fixpoints are
required in the $\mu$-calculus (see, e.g.,~\cite{GrTW02}).  Forms of
$\mu$-calculus have also been considered in Description
Logics~\cite{DeLe94b,Schi94,KuSV02,BLMV08}, again interpreted over infinite
trees.  In this context, the present work can provide the foundations for a
description logic tailored towards acyclic finite (a.k.a.\ well-founded) frame
structures.  In this sense, the present work overcomes~\cite{CaDL99c}, where an
explicit well-foundedness construct was used to capture XML in description
logics.

In a finite-tree setting, extending \XPath with various forms of fixpoint
operators
%MYV: Added BaLi05
has been studied earlier~\cite{AGMRT08,BaLi05,Tenc06,GeLS07,GeLa10,Libk06b}.
While for arbitrary fixpoints the resulting query language is equivalent
to MSO and has an exponential-time containment test, it is not known to have a
linear-time evaluation algorithm.  In contrast, as \muXPath is alternation free
it is closely related to a stratified version of Monadic Datalog proposed as a
query language for finite trees in~\cite{GoKo04,FrGK03}, which enjoys
linear-time evaluation.  Note, however, that the complexity of containment of
stratified Monadic Datalog is unknown,
%MYV: Added
though the techniques developed here may be useful in settling that open
question.

We prove here that there is a very direct correspondence between \muXPath and
2WATAs. Specifically, there are effective translations from \muXPath queries to
2WATAs and from 2WATAs to \muXPath.  We show that this yields the nice
computational properties for \muXPath.  We then prove the equivalence of 2WATAs
to node-selecting tree automata (NSTA), shown to be expressively equivalent to
MSO~\cite{FrGK03}. On the one hand, we have an exponential translation from
2WATAs to NSTAs. On the other hand, we have a linear translation from NSTAs to
2WATAs.  This yields the expressive equivalence of \muXPath to MSO.

It is worth noting that the automata-theoretic approach of 2WATAs is based on
techniques developed in the past 20 years in the context of program
logics~\cite{KuVW00,Vard98}.  Here, however, we leverage the fact that we are
dealing with \emph{finite} trees, rather than \emph{infinite} trees that are
usually used in the program-logics context.  Indeed, the automata-theoretic
techniques used in reasoning about infinite trees are notoriously difficult and
have resisted efficient
%MYV: Added a reference.
implementation~\cite{FrLa10,ScTW05,TaHB95}.  The restriction to finite trees
here enables one to obtain a much more feasible algorithmic approach.  In
particular, one can make use of symbolic techniques, at the base of modern
model checking tools~\cite{BCMDH92}, for effectively querying and verifying XML
documents.  It is worth noting that while 2WATAs run over finite trees they are
allowed to have infinite runs.  This separates 2WATAs from the alternating
finite-tree automata used elsewhere~\cite{CGKV88,Slut85}.
%MYV: Added
Thus, the technical contribution here is not the development of completely new
automata-theoretic techniques, but the adaptation and simplification of known
techniques to the setting of finite-trees.
% To be helpful to readers with deep expertise in automata theory,
The constructions and algorithms for this adaptation are presented here
% for full details,
with minimal references to prior automata-theoretic work.

The paper is organized as follows. In Section~\ref{sec:muxpath} we present
syntax and semantics of \muXPath, as well as examples of queries expressed in
this language. In Sections~\ref{sec:muXPath-2WATA} and~\ref{sec:2wata} we show
how to make use of two-way automata over finite trees as a formal tool for
addressing query evalation and query containemnt in the context of
\muXPath. More specifically, in Section~\ref{sec:muXPath-2WATA} we introduce
the class of two-way weak alternating tree automata, and devise mutual
translation between them and \muXPath queries, and in Section~\ref{sec:2wata}
we provide algorithms for deciding the acceptance and non-emptiness problems
for 2WATAs.  Section~\ref{sec:2wata-mso} deals with the expressive power of
\muXPath, by establishing the relationship between two-way weak alternating
tree automata and MSO.  In Section~\ref{sec:reasoning} we exploit the
correspondence between two-way weak alternating tree automata and \muXPath to
illustrate the main characteristics of \muXPath as a query language over finite
trees.  Finally, Section~\ref{sec:conclusion} concludes the paper.

%%% Local Variables:
%%% mode: latex
%%% TeX-master: "main"
%%% End:

\section{The Query Language \muXPath}
\label{sec:muxpath}

In this paper we are concerned with query languages over tree-structured data,
which is customary in the XML \cite{Marx04b,Marx05} and more recently
JSON \cite{BRSV17,HiPV17,BCCX18} format.  More precisely, we consider databases
as finite sibling-trees, which are tree structures whose nodes are linked to
each other by two relations: the child relation, connecting each node with its
children in the tree; and the immediate-right-sibling relation, connecting each
node with its sibling immediately to the right in the tree.  Such a relation
models the order between the children of the node in a %an XML
document.
Each node of the sibling tree is labeled by (possibly many) elements of a fixed set
$\Sigma$ of atomic propositions
% We consider the set $\Sigma$ to be partitioned into $\SigmaA$ and
% $\SigmaId$. The set $\SigmaA$ is a set of atomic propositions
that represent either %XML 
tags or %XML 
attribute-value pairs.
% Instead, $\SigmaId$ is a set of special propositions representing (node)
% \emph{identifiers}, i.e., that are true in (i.e., that label) exactly a
% single node of the XML document.  Such identifiers are essentially an
% abstraction of the XML identifiers built through the construct \texttt{ID}
% (see footnote~\ref{footnote-xml}), though a node can have multiple
% identifiers in our case.
Observe that in general sibling trees are more general than XML documents since
they would allow the same node to be labeled by several tags.

Formally, a (finite) tree is a complete prefix-closed non-empty (finite) set of
words over $\Nat_+$, i.e., the set of positive natural numbers.  In other words,
a \emph{(finite) tree} is a (finite) set of words $\Delta\incl\Nat_+^*$, such
that if $x\conc i\in\Delta$, where $x\in\Nat^*$ and $i\in\Nat_+$, then also
$x\in\Delta$.  The elements
of $\Delta$ are called \emph{nodes}, the empty word $\eword$ is the \emph{root}
of $\Delta$, and for every $x\in\Delta$, the nodes $x\conc i$, with $i\in\Nat$,
are the \emph{children} of $x$.  By convention we take $x\conc 0=x$, and
$x\conc i\conc{-1}=x$.  The \emph{branching degree} $d(x)$ of a node $x$
denotes the number of children of $x$.  If the branching degree of all nodes
of a tree is bounded by $k$, we say that the tree is \emph{ranked} and has
branching degree $k$.  In particular, if the branching degree is~2, we say that
the tree is \emph{binary}.  Instead, if the number of children of the nodes
is a priori unbounded, we say that the tree is \emph{unranked}.
A \emph{(finite) labeled tree} over an alphabet $\L$ of labels is a pair
$T=\Tree[T]$, where $\Dom[T]$ is a (finite) tree and the labeling
$\ell^T:\Dom[T]\ra\L$ is a mapping assigning to each node $x\in\Dom[T]$ a label
$\Lab{x}$ in $\L$.

A \emph{sibling tree} $T$ is a finite labeled unranked tree $\Delta$
with the additional requirement that if node $x\conc (i+1)\in\Delta$,
then also node $x\conc (i)\in\Delta$.Moreover, nodes of sibling trees
are labeled with a set of atomic propositions in an alphabet $\Sigma$,
i.e., $\L=2^{\Sigma}$.  Given $A\in\Sigma$, we denote by $\Int[T]{A}$
the set of nodes $x$ of $\Dom[T]$ such that $A\in\Lab{x}$.  It is
customary to denote a sibling tree $T$ by $(\Dom[T],\Int[T]{\cdot})$.
On sibling trees, two auxiliary binary relations between nodes, and their
inverses are defined ($i>0$):
\[
  \begin{array}[t]{lcl}
    \Int[T]{\achild}  &=& \{(z,z\conc i) \mid z, z\conc i\in\Dom[T]\}\\
    \Int[T]{(\achild^-)}  &=& \{(z\conc i,z) \mid z, z\conc i\in\Dom[T]\}\\
    \Int[T]{\aright}  &=&
      \{(z\conc i,z\conc (i{+}1)) \mid z\conc i, z\conc (i{+}1)\in\Dom[T]\}\\
    \Int[T]{(\aright^-)}  &=&
      \{(z\conc (i{+}1),z\conc i) \mid z\conc i, z\conc (i{+}1)\in\Dom[T]\}
    % \INT[T]{P^-} &=& \{(z',z) \mid (z,z')\in\Int[T]{P}\}
  \end{array}
\]
The relations \achild and \aright are called \emph{axes}.

\medskip

One of the core languages used to query tree-structured data is \XPath, whose
definition we briefly recall here.  An \XPath \emph{node expression} $\varphi$
is defined by the following syntax, which is inspired by Propositional Dynamic
Logic (PDL) \cite{FiLa79,ABDG*05,CDLV09}:
\begin{eqnarray*}
  \varphi &\lora& A                         \mid
                  \lnot \varphi             \mid
                  \varphi_1 \land \varphi_2 \mid
                  \varphi_1 \lor \varphi_2  \mid
                  \DIAM{P}{\varphi}         \mid
                  \BOX{P}{\varphi}          \\
  P &\lora&  \achild        \mid
             \aright        \mid
             \achild^-      \mid
             \aright^-
             % \TEST{\varphi} \mid
             % P_1; P_2       \mid
             % P_1 \cup P_2   \mid
             % P^*            \mid
             % P^-
\end{eqnarray*}
where $A$ denotes an atomic proposition belonging to an alphabet $\Sigma$,
$\achild$ and $\aright$ denote the main atomic relations between nodes in a
tree, usually called \emph{axis} relations.  The expressions $\achild^-$ and
$\aright^-$ denote their inverses, which in fact correspond to the other two
standard \XPath axes $\apar$ and $\aleft$, respectively.  Intuitively, a node
expression is a formula specifying a property of nodes, where an atomic
proposition $A$ asserts that the node is labeled with $A$, negation,
conjunction, and disjunction have the usual meaning, $\DIAM{P}{\varphi}$, where
$P$ is one of the axes, denotes that the node is connected via $P$ with a node
satisfying $\varphi$, and $\BOX{P}{\varphi}$ asserts that all nodes connected
via $P$ satisfy $\varphi$.  We also adopt the usual abbreviations for booleans,
i.e., $\TRUE$, $\FALSE$, and $\varphi_1\limp\varphi_2$.

The query language studied in this paper, called \muXPath is an extension of
\XPath with a mechanism for defining sets of nodes by means of explicit
fixpoint operators over systems of equations.  \muXPath is essentially the
Alternation-Free $\mu$-Calculus, where the syntax allows for the fixpoints to
be defined over vectors of variables~\cite{EmLe86}.

To define \muXPath queries, we consider a set $\X$ of variables, disjoint from
the alphabet $\Sigma$.  An \emph{equation} has the form
\[
  X \doteq \varphi
\]
where $X\in\X$, and $\varphi$ is an \XPath node expression having as atomic
propositions symbols from $\Sigma\cup\X$.  We call the left-hand side of the
equation its \emph{head}, and the right-hand side its \emph{body}.
A set of equations can be considered as mutual fixpoint equations, which can
have multiple solutions in general. We are actually interested in two
particular solutions: the smallest one, i.e., the least fixpoint (lfp), and the
greatest one, i.e., the greatest fixpoint (gfp), both of which are guaranteed
to exist under a suitable syntactic monotonicity condition to be defined below.
Given a set of equations
\[
  \{X_1\doteq\varphi_1,\dots,X_n\doteq\varphi_n\},
\]
where we have one equation with $X_i$ in the head, for $1\leq i\leq n$, a
\emph{fixpoint block} has the form
$\mathit{fp}{\{X_1\doteq\varphi_1,\dots,X_n\doteq\varphi_n\}}$, where
\begin{itemize}
\item $\mathit{fp}$ is either $\lfp$ or $\gfp$, denoting respectively the least
  fixpoint and the greatest fixpoint of the set of equations, and
\item each variable $X_i$, for $1\leq i\leq n$, appears \emph{positively} in
  $\varphi_i$, i.e., it appears within the scope of an even number of negations.
\end{itemize}
The latter condition, called \emph{syntactic monotonicity}, syntactically
guarantees monotonicity, and hence, the existence of least and greatest
fixpoints (see \cite{Tars55,Koze83}).  We say that the variables
$X_1,\dots,X_n$ are \emph{defined} in the fixpoint block
$\mathit{fp}{\{X_1\doteq\varphi_1,\dots,X_n\doteq\varphi_n\}}$.

A \emph{\muXPath query} has the form $X:\F$, where $X\in\X$ and $\F$ is a set
of fixpoint blocks such that:
\begin{itemize}
%\item $X$ is a variable defined (one of the fixpoint blocks) in $\F$;
\item each variables occurring in $X:\F$, including $X$ is defined in exactly one  of the fixpoint blocks of $\F$;
% \item the sets of variables defined in different fixpoint blocks in $\F$ are
%   mutually disjoint;
% \item for each fixpoint block $F\in\F$, each variable $X$ defined in $F$
%   appears only positively in the bodies of equations in $F$ (\emph{syntactic
%    monotonicity});
\item there exists a partial order $\preceq$ on the fixpoint blocks in $\F$
  such that, for each $F_i\in\F$, the bodies of equations in $F_i$ contain only
  variables defined in fixpoint blocks $F_j\in\F$ with $F_j\preceq F_i$.
\end{itemize}

The meaning of a query $q$ of the form $X:\F$ is based on the fact that, when
evaluated over a tree $T$, $\F$ assigns to each variable defined in it a set of
nodes of $T$, and that $q$ returns as result the set assigned to $X$.  We
intuitively explain the mechanism behind the assignment of $\F$ to its
variables.  We choose partial order $\preceq$ on the fixpoint blocks in $\F$
respecting the conditions above, and we operate one block of equations at a
time according to $\preceq$.
For each fixpoint block, we compute the solution of the corresponding
equations, obviously taking into account the type of fixpoint, and using the
assignments for the variables already computed for previous blocks.
We come back to the formal semantics below, and first give some examples of
\muXPath queries.

The following query computes the nodes reaching a $\ared$ node on all
$\achild$-paths (possibly of length 0), exploiting the encoding of transitive
closure by means of a least fixpoint:
\[
  X:\{\LFP{X\doteq \ared \lor \DIAM{\achild}{\TRUE}\land\BOX{\achild}{X}}\}.
\]

As another example, to obtain the nodes such that all their descendants (including the
node itself) are not simultaneously $\ared$ and $\ablue$, we can write the
query:
\[
  X:\{\GFP{X\doteq (\ared \limp \lnot\ablue) \land \BOX{\achild}{X}}\}.
\]
Notice that such nodes are those that do not have descendants that are
simultaneously $\ared$ and $\ablue$.  The latter set of nodes is characterized
by a least fixpoint, and therefore query $q$ can also be considered as the
negation of such least fixpoint.

We now illustrate an example where both a least and a greatest
fixpoint block are used in the same query.  Indeed, to compute $\ared$
nodes such that all their $\ared$ descendants have only $\ablue$
children (if any) and such that all their $\ablue$ descendants have at
least a $\ared$ child, we can use the following query:
\[
  X_1:\{
  \begin{array}[t]{@{}l}
    \gfp\{X_0 \doteq
       \begin{array}[t]{l}
         (\ared\limp \BOX{\achild}{\ablue}) \land {}\\
         (\ablue\limp \DIAM{\achild}{\ared}) \land \BOX{\achild}{X_0} \}\},
       \end{array}\\
    \LFP{X_1 \doteq \ared\land X_0}.
  \end{array}
\]
Notice that in the above query, the only partial order coherent with the
conditions of \muXPath given above is the one where the greatest fixpoint block
precedes the least fixpoint block.

Notice also that in the above query we could have used the greatest fixpoint in
the second block instead of the least fixpoint.  Indeed, it is easy to see
that, whenever a set of equations in non-recursive, least and greatest
fixpoints have the same meaning, since they both characterize the obvious
single solution of the systems of equations.

Now, suppose that we want to denote the $\ared$ nodes such that all
their $\ared$ descendants reach $\ablue$ nodes on all $\achild$-paths,
and such that all their $\ablue$ descendants reach $\ared$ nodes on at
least one $\achild$-path.  The resulting query is the following, where
we have written the fixpoint blocks according to a partial order
coherent with the conditions of \muXPath:
\[
  X_3: \{
  \begin{array}[t]{@{}l}
    \LFP{X_0\doteq \ablue \lor \BOX{\achild}{X_0}},\\
    \LFP{X_1\doteq \ared \lor \DIAM{\achild}{X_1}},\\
    \GFP{X_2 \doteq (\ared\limp X_0) \land (\ablue\limp X_1) \land
     \BOX{\achild}{X_2}},\\
    \LFP{X_3 \doteq \ared\land X_2}.
    \}
  \end{array}
\]

Finally, to denote the nodes having a $\ared$ sibling that follows it in the
sequence of right siblings, and such that all siblings along such sequence have
a $\ablue$ descendant, we can use the following query:
\[
  X_0: \{
  \lfp\{
  \begin{array}[t]{@{}l}
    X_0\doteq X_1 \land (\ared \lor \DIAM{\aright}{X_0}),\\
    X_1\doteq \ablue \lor  \DIAM{\achild}{X_1} \}\}.
  \end{array}
\]

\begin{figure}[tb]
\[
  \begin{array}{@{}l}%[t]{@{}l@{}}
    \Intval{T}{\rho}{A} = \Int[T]{A}, \\
    \Intval{T}{\rho}{X} = \rho(X),\\
    % \begin{cases}
    %   \rho(X), & \text{if $X$ is defined in $F_i$}\\
    %   \E, & \text{if $X$ is defined in some $F_j\preceq F_i$ and
    %    $X/\E\in\INTVAL{T}{\rho}{F_j}$}
    % \end{cases}\\
    \INTVAL{T}{\rho}{\lnot\varphi} =
      \Dom \setminus \Intval{T}{\rho}{\varphi},\\
    \INTVAL{T}{\rho}{\varphi_1\land\varphi_2} =
    \INTVAL{T}{\rho}{\varphi_1}\cap\INTVAL{T}{\rho}{\varphi_2},\\
    \INTVAL{T}{\rho}{\varphi_1\lor\varphi_2} =
    \INTVAL{T}{\rho}{\varphi_1}\cup\INTVAL{T}{\rho}{\varphi_2},\\
    \INTVAL{T}{\rho}{\DIAM{P}{\varphi}} = \{z \mid \exists z'.
      (z,z')\in \Int[T]{P} \land z'\in\Intval{T}{\rho}{\varphi}\},\\
    \INTVAL{T}{\rho}{\BOX{P}{\varphi}} = \{z \mid \forall z'.
      (z,z')\in \Int[T]{P} \limp z'\in\Intval{T}{\rho}{\varphi}\}.
    % \INTVAL{T}{\rho}{
    %    \LFP{X_1 \doteq \varphi_1,\dots, X_n \doteq \varphi_n}} =
    %   \{X_1/\E_1^\mu,\dots,X_n/\E_n^\mu\},\\
    % \INTVAL{T}{\rho}{
    %    \GFP{X_1 \doteq \varphi_1,\dots, X_n \doteq \varphi_n}} =
    %   \{X_1/\E_1^\nu,\dots,X_n/\E_n^\nu\},
  \end{array}
\]
  \caption{Semantics of \XPath node expression in \muXPath}
  \label{fig:semantics-ne}
\end{figure}

The formal semantics of \muXPath is defined by considering sibling trees as
interpretation structures.
In addition, we need  second order variable
assignments to interpret variables.  A \emph{(second order) variable assignment} $\rho$ on a tree
$T=(\Dom[T],\Int[T]{\cdot})$ is a mapping that assigns to variables of $\X$
sets of nodes in $\Dom[T]$.
With this notion we can now interpret \XPath node expression extended with second-order variables used in the body of equations as in
Figure~\ref{fig:semantics-ne}.
Now we turn to \muXPath  fixpoint blocks.
%\begin{itemize}
% \item The semantics of $A$, $\lnot\varphi$, $\varphi_1\land\varphi_2$,
%   $\varphi_1\lor\varphi_2$, $\DIAM{P}{\varphi}$, and $\BOX{P}{\varphi}$ is the
%   usual one.
% \item The semantics of a variable $X$ depends on whether $X$ is defined in
%   $F_i$ or not. In the former case, it is simply given by the variable
%   assignment $\rho$; otherwise, it is determined by the variable assignment of
%   block $F_j$ in which $X$ is defined.  Observe that $F_j$ precedes $F_i$ in
%   the partial order $\preceq$.
%\item 
The semantics of a least fixpoint block
$\LFP{X_1\doteq\varphi_1,\dots, X_n\doteq\varphi_n}$ is the variable assignment
$\{X_1/\E_1^\mu,\dots,X_n/\E_n^\mu\}$, where $(\E_1^\mu,\dots,\E_n^\mu)$ is the
intersection of all solutions of the fixpoint block \cite{Tars55}. Note that
each solution is an $n$-tuple of sets of nodes of $T$, and the intersection is
done component-wise.  Formally:
  \[
    (\E_1^\mu,\dots,\E_n^\mu) ~=~
    \bigcap \{ (\E_1,\dots,\E_n) \mid \E_1 =
    \INTVAL{T}{\rho[X_1/\E_1,\ldots,X_n/\E_n]}{\varphi_1}, \ldots, \E_n =
    \INTVAL{T}{\rho[X_1/\E_1,\ldots,X_n/\E_n]}{\varphi_n}\},
  \]
  where $\rho[X_1/\E_1,\dots,X_n/\E_n]$ denotes the variable assignment
  identical to $\rho$, except that it assigns to $X_i$ the value $\E_i$, for
  $1\leq i\leq n$.
  Due to syntactic monotonicity, $(\E_1^\mu,\dots,\E_n^\mu)$ is itself a
  solution of the fixpoint block, and indeed the smallest one
  \cite{Tars55,Koze83}.

%\item 
The semantics of a greatest fixpoint block
  $\GFP{X_1\doteq\varphi_1,\dots, X_n\doteq\varphi_n}$ is the variable
  assignment $\{X_1/\E_1^\nu,\dots,X_n/\E_n^\nu\}$, where
  $(\E_1^\nu,\dots,\E_n^\nu)$ is the union of all solutions of the fixpoint
  block, i.e.:
  \[
    (\E_1^\nu,\dots,\E_n^\nu) ~=~
    \bigcup \{ (\E_1,\dots,\E_n) \mid \E_1 =
    \INTVAL{T}{\rho[X_1/\E_1,\ldots,X_n/\E_n]}{\varphi_1}, \ldots, \E_n =
    \INTVAL{T}{\rho[X_1/\E_1,\ldots,X_n/\E_n]}{\varphi_n}\}.
  \]
  Again due to syntactic monotonicity, $(\E_1^\nu,\dots,\E_n^\nu)$
  is itself a solution of the fixpoint block, and in this case the largest one.
%\end{itemize}

To define the \emph{semantics} of a \muXPath query $X:\F$ relative to a sibling tree
$T$ and a variable assignment $\rho$, we consider a partial order $\preceq$ of
the fixpoint blocks in $\F$, and proceed by induction on $\preceq$.  
Each block $F_i \in \F$ provides a variable
assignment $\{X_1/\E_1,\dots,X_n/\E_n\}$ for the variables  $X_1,\ldots,X_n$ defined in $F_i$, where $\E_1,\ldots,\E_n$ are the sets of nodes of $T$
associated to such variables by the fixpoint computation. 
%   $F_i$ or not. In the former case, it is simply given by the variable
%   assignment $\rho$; otherwise, it is determined by the variable assignment of
%   block $F_j$ in which $X$ is defined.  Observe that $F_j$ precedes $F_i$ in
%   the partial order $\preceq$.
In particular, 
$X:\F$ over a sibling tree
$T$ is the set $\E\subseteq\Dom[T]$ of nodes of $T$ that the fixpoint block
$F\in\F$ defining $X$ assigns to $X$ in $T$.  We denote such set $\E$ as
$\INT[T]{X:\F}$.  Notice that, since all second-order variables
 appearing in $\F$ are assigned values in the fixpoint block in which they are
 defined, we can omit from $\INTVAL{T}{\rho}{X:\F}$ the second order variables
 assignment $\rho$, and denote it as $\INT[T]{X:\F}$.

\medskip

Although syntactically different, \muXPath  can be seen as an extension of \RXPath
 \cite{Marx04b,Marx05,CDLV09}. In \RXPath node expressions $\DIAM{P}\phi$ and $\BOX{P}{\phi}$ allow for paths  $P$ that are regular expressions over the \XPath axes, namely:
\[
  P \lora  \achild        \mid
             \aright        \mid
             \TEST{\varphi} \mid
             P_1; P_2       \mid
             P_1 \cup P_2   \mid
             P^*            \mid
             P^-.
\]
In \muXPath such \RXPath node expressions can be considered syntactic sugar,
and added at no cost.

To see this consider the following.  Starting from \cite{Koze83}, but also in
\cite{CoMe90,GoKo04}, it has been observed several times in the literature that
node expression of the form $\DIAM{P}\phi$ and $\BOX{P}{\phi}$ with complex $P$
can be considered as abbreviations for suitable fixpoint expressions. This can
be easily done in \muXPath as well.

First of all, we notice that in expressions of the form $P^-$, we can apply
recursively the following equivalences to push the inverse operator $^-$ inside
\RXPath expressions, until it is applied to \XPath axes only:
\[
  \begin{array}{rcl}
    (\TEST{\varphi})^- &=& \TEST{\varphi},\\
    (P_1; P_2)^- &=& P_2^-;P_1^-,\\
    (P_1 \cup P_2)^- &=& P_1^-\cup P_2^-,\\
    (P^*)^- &=& (P^-)^*.
  \end{array}
\]
Also, considering that
$\varphi_1\lor\varphi_2\equiv\lnot(\lnot\varphi_1\land\lnot\varphi_2)$, and
$\BOX{P}{\varphi}\equiv\lnot\DIAM{P}{\lnot\varphi}$, we can assume w.l.o.g.,
that \RXPath queries are formed as follows:
\[
  \begin{array}{rcl}
    \varphi &\lora& A                         \mid
                    \lnot \varphi             \mid
                    \varphi_1 \land \varphi_2 \mid
                    \DIAM{P}{\varphi} ,       \\
    P &\lora&  \achild        \mid
               \aright        \mid
               \achild^-      \mid
               \aright^-      \mid
               \TEST{\varphi} \mid
               P_1; P_2       \mid
               P_1 \cup P_2   \mid
               P^*.
  \end{array}
\]
% Also, we assume that \RXPath queries are in negation normal form, by pushing
% negation inwards, exploiting De Morgan's law and the equivalence
% $\BOX{P}{\varphi}\equiv\lnot\DIAM{P}{\lnot\varphi}$.
%
Then, we can transform an arbitrary \RXPath query $\varphi$ into the \muXPath
query $X_{\varphi}:\F$, where $\F$ is a set of fixpoint blocks constructed by
inductively decomposing $\varphi$.  Formally, we let $\F=\tau(\varphi)$, where
$\tau(\varphi)$ is defined by induction on $\varphi$ as follows:
\[
  \begin{array}{rcl}
    \tau(A) &=& \{ \LFP{X_A\doteq A} \},\\
    \tau(\lnot\varphi') &=&
      \{\LFP{X_{\lnot\varphi'}\doteq\lnot X_{\varphi'}}\} \cup \tau(\varphi'),\\
    \tau(\varphi_1\land\varphi_2) &=&
      \{\LFP{X_{\varphi_1\land\varphi_2}\doteq X_{\varphi_1}\land X_{\varphi_2}}\}
      \cup \tau(\varphi_1) \cup \tau(\varphi_2),\\
    \tau(\DIAM{P}{\varphi'}) &=&
      \{\lfp\,\tau_p(\DIAM{P}{\varphi'})\} \cup \tau_t(P) \cup \tau(\varphi'),
  \end{array}
\]
where the function $\tau_p(\cdot)$, defined over formulas of the form
$\DIAM{P}{\varphi'}$, returns a set of fixpoint equations, and the function
$\tau_t(\cdot)$, defined over path expressions $P$, returns the set of
fixpoint blocks corresponding to the node formulas appearing in the tests in
$P$.  In particular, $\tau_p(\DIAM{P}{\varphi'})$ is defined by induction on the
structure of the path expression $P$ as follows:
\[
  \begin{array}{rcl}
    \tau_p(\DIAM{\mathit{axis}}{\varphi'}) &=&
      \{X_{\DIAM{\mathit{axis}}{\varphi'}} \doteq
        \DIAM{\mathit{axis}}{X_{\varphi'}}\},
    \qquad\text{for } \mathit{axis}\in\{\achild,\aright,\achild^-,\aright^-\},\\
    \tau_p(\DIAM{\TEST{\varphi''}}{\varphi'}) &=&
      \{X_{\DIAM{\TEST{\varphi''}}{\varphi'}} \doteq X_{\varphi''}\land X_{\varphi'}\},\\
    \tau_p(\DIAM{P_1;P_2}{\varphi'}) &=&
      \{X_{\DIAM{P_1;P_2}{\varphi'}}\doteq X_{\DIAM{P_1}{\DIAM{P_2}{\varphi'}}}\} \cup
      \tau_p(\DIAM{P_1}{\DIAM{P_2}{\varphi'}}) \cup
      \tau_p(\DIAM{P_2}{\varphi'}),\\
    \tau_p(\DIAM{P_1\cup P_2}{\varphi'}) &=&
      \{X_{\DIAM{P_1\cup P_2}{\varphi'}} \doteq
      X_{\DIAM{P_1}{\varphi'}} \lor X_{\DIAM{P_2}{\varphi'}}\} \cup
      \tau_p(\DIAM{P_1}{\varphi'}) \cup \tau_p(\DIAM{P_2}{\varphi'}),\\
    \tau_p(\DIAM{P^*}{\varphi'}) &=&
      \{X_{\DIAM{P^*}{\varphi'}}\doteq X_{\varphi'} \lor
        \DIAM{P}{X_{\DIAM{P^*}{\varphi'}}} \} \cup \tau_p(\DIAM{P}{\varphi'}).
  \end{array}
\]
Note that $\tau_p$ decomposes inductively only the path expression inside the
first $\DIAM{\cdot}{}$ formula.  Hence, the \muXPath formula $\tau(\varphi)$ is
linear in the size of the \RXPath formula $\varphi$.

For example $\DIAM{\aright^*}{A}$ can be expressed as
\[
  X:\{\LFP{X\doteq A\lor\DIAM{\aright}{X}}\}.
\]
Instead, $\BOX{\aright^*}{A}$, which is equivalent to
$\lnot\DIAM{\aright^*}{\lnot A}$, can be expressed as
\begin{equation}
  \label{eqn:gfp-example}
  X:\{
  \begin{array}[t]{@{}l}
    \LFP{X\doteq \lnot X_1},\\
    \LFP{X_1\doteq \lnot A\lor\DIAM{\aright}{X_1}}\},
  \end{array}
\end{equation}
which in turn is equivalent to
\[
  X:\{\GFP{X\doteq A\land\BOX{\achild}{X}}\}.
\]
Observe that the form of equation~\eqref{eqn:gfp-example} resembles the
encoding of the corresponding \RXPath formula into stratified Monadic
Datalog~\cite{GoKo04}.

% Notice that in \muXPath we do not consider explicitly identifiers, i.e.,
% atomic propositions that denote singletons.  Indeed in \RXPath (and hence in
% \muXPath) we can express that a proposition $A$ denotes a singleton, by
% requiring that the following expression holds at the root:
% \[
%   \begin{array}{l}
%     \DIAM{\achild^*}{ A}\land {}
%     \BOX{\achild^*} (A \limp
%     (\begin{array}[t]{@{}l}
%       (\BOX{\achild^+}{\lnot A}) \land{}\\
%       \BOX{(\achild^-)^+}
%       {(\lnot A\land\BOX{(\aright^+\cup(\aright^-)^+);\achild^*}{\lnot A})}))
%     \end{array}
%   \end{array}
% \]
% where $P^+$ abbreviates $P;P^*$.

\medskip

We close this section by observing that although sibling trees are unranked, we can transform  them into \emph{binary sibling trees}
by considering an additional axis $\afchild$, connecting each node to its first
child only, interpreted as
\[
  \Int[T]{\afchild}  ~=~ \{(z,z\conc 1) \mid z, z\conc 1\in\Dom[T]\}.
\]
Using $\afchild$, we can thus re-express the $\achild$ axis as
$\afchild;\aright^*$.
In the following, we will focus on \muXPath queries that use only the
$\afchild$ and $\aright$ axis relations, and are evaluated over binary sibling
trees.

%%% Local Variables:
%%% mode: latex
%%% TeX-master: "main"
%%% End:

% LocalWords:  rxpath PDL XML XPath wrt

%%%%%%%%%%%%%%%%%%%%%%%%%%%%%%%%%%%%%%%%%%%%%%%%%%%%%%%%%%%%%%%%%%%%%%%%%%%%%%
%%% Time-stamp: "2018-11-09 18:42:42 calvanese"
%%%%%%%%%%%%%%%%%%%%%%%%%%%%%%%%%%%%%%%%%%%%%%%%%%%%%%%%%%%%%%%%%%%%%%%%%%%%%%

\section{2WATAs and their Relationship to \muXPath}
\label{sec:muXPath-2WATA}

We consider now two-way automata over finite trees and use them as a formal
tool to address the problems about \muXPath in which we are interested in this
paper.  Specifically, after having introduced the class of two-way weak
alternating tree automata (2WATAs), we establish a tight relationship between
them and \muXPath by devising mutual translations between the two formalisms.

\subsection{Two-way Weak Alternating Tree Automata}
\label{sec:2ata}

We consider a variant of two-way alternating automata~\cite{Slut85} (see
also~\cite{Neve02,CDGJ*08}) that run, possibly infinitely, on finite labeled
trees (Note that typically, infinite runs of automata are considered in the
context of infinite input structures~\cite{GrTW02}, whereas here we consider
possibly infinite runs over finite structures.)  Specifically, alternating tree
automata generalize nondeterministic tree automata, and two-way tree automata
generalize ordinary tree automata by being allowed to traverse the tree both
upwards and downwards.  Formally, let $\B^+(I)$ be the set of positive Boolean
formulae over a set $I$, built inductively by applying $\land$ and $\lor$
starting from $\true$, $\false$, and elements of $I$.  For a set $J\incl I$ and
a formula $\varphi\in\B^+(I)$, we say that $J$ \emph{satisfies} $\varphi$ if
assigning $\true$ to the elements in $J$ and $\false$ to those in
$I\setminus J$, makes $\varphi$ true. We make use of $[-1..k]$ to denote
$\{{-1},0,1,\ldots,k\}$, where $k$ is a positive integer.
A \emph{two-way weak alternating tree automaton} (2WATA) running over labeled
trees
% with branching degree $k$,
all of whose nodes have at most $k$ children, is a
tuple $\At=(\L,S,s_0,\delta,\alpha)$, where $\L$ is the alphabet of tree
labels, $S$ is a finite set of states, $s_0\in S$ is the initial state,
$\delta:S\times\L\ra\B^+([-1..k]\times S)$ is the transition function, and
$\alpha$ is the accepting condition discussed below.

The transition function maps a state $s\in S$ and an input label $a\in\L$ to a
positive Boolean formula over $[-1..k]\times S$.  Intuitively, if
$\delta(s,a)=\varphi$, then each pair $(c',s')$ appearing in $\varphi$
corresponds to a new copy of the automaton going to the direction suggested by
$c'$ and starting in state $s'$.  For example, if $k=2$ and
$\delta(s_1,a)=((1,s_2)\land(1,s_3))\lor((-1,s_1)\land(0,s_3))$, when the
automaton is in the state $s_1$ and is reading the node $x$ labeled by $a$, it
proceeds either by sending off two copies, in the states $s_2$ and $s_3$
respectively, to the first child of $x$ (i.e.,
$x\conc 1$), or by sending off one copy in the state $s_1$ to the
parent of $x$ (i.e., $x\conc{-1}$) and one copy in
the state $s_3$ to $x$ itself (i.e., $x\conc 0$).
%% Moreover, the integer $c$ represents the fact that the automaton can check
%% for the presence or absence of a child of the current node.  If $c>0$, then
%% the automaton can perform the transition only if the current node \emph{has}
%% child $c$.  Similarly, if $c<0$, then the automaton can perform the
%% transition only if the current node does \emph{not} have child $c$.  If
%% $c=0$ then the automaton can always perform the transition.
%% Note that to perform a transition only if the current node \emph{does have}
%% child $c$, it is sufficient to add a conjunct $(c,s')$ to $\varphi$.

A run of a 2WATA is obtained by resolving all existential choices. The
universal choices are left, which gives us a tree. Because we are
considering two-way automata, runs can start at arbitrary tree nodes,
and need not start at the root.
Formally, a \emph{run} of a 2WATA $\At$ over a labeled tree $T=\Tree[T]$ from a
node $x_0\in\Delta^T$ is, in general, an infinite $\Dom\times S$-labeled tree
$R=\Tree[R]$ satisfying:
\begin{enumerate}
\item $\eword\in\Dom[R]$ and $\Lab[R]{\eword}=(x_0,s_0)$.
\item
Let $\Lab[R]{r}=(x,s)$ and $\delta(s,\Lab[T]{x})=\varphi$.
Then there is a (possibly empty) set
$P=\{(c_1,s_1),\ldots,(c_n,s_n)\}\incl[-1..k]\times S$ of pairs such that
% \begin{inparaitem}[]
% \item
$P$ satisfies $\varphi$, and
% \item
for each $i\in\{1,\ldots,n\}$, we have that $r\conc i\in\Dom[R]$,
$x\conc c_i\in\Dom[T]$, and $\Lab[R]{r\conc i}=(x\conc c_i,s_i)$.
%\end{inparaitem}
  In particular, this means that $\varphi$ cannot be $\false$, and if $\varphi$
  is $\true$ then $r$ need not have children.
\end{enumerate}
Intuitively, a run $R$ keeps track of all transitions that the 2WATA $\At$
performs on a labeled input tree $T$: a node $r$ of $R$ labeled by $(x,s)$
describes a copy of $\At$ that is in the state $s$ and is reading the node $x$
of $T$.  The children of $r$ in the run represent
the transitions made by the multiple copies of $\At$ that are being sent off
either upwards to the parent of $x$, downwards to
one of the children of $x$, or to $x$ itself.

%Observe that a 2WATA may fail to have a run on a labeled tree, and this happens
%if on a certain node the transition function is $\false$, and hence cannot be
%satisfied.  Also, runs may be finite, i.e., they may contain finite branches
%e.g., when on a certain node the transition function is $\true$.

2WATAs are called ``weak'' due to the specific form of the acceptance
condition, given in the form of a set $\alpha\subseteq S$~\cite{KuVW00}.
Specifically, there exists a partition of $S$ into disjoint sets, $S_i$, such
that for each set $S_i$, either $S_i\subseteq \alpha$, in which case $S_i$ is
an \emph{accepting set}, or $S_i\cap \alpha=\emptyset$, in which case $S_i$ is
a \emph{rejecting set}.  In addition, there exists a partial order $\leq$ on
the collection of the $S_i$'s such that, for each $s\in S_i$ and $s'\in S_j$
for which $s'$ occurs in $\delta(s,a)$, for some $a\in\L$, we have $S_j\leq
S_i$.  Thus, transitions from a state in $S_i$ lead to states in either the
same $S_i$ or a lower one.  It follows that every infinite path of a run of a
2WATA ultimately gets ``trapped'' within some $S_i$.  The path is
\emph{accepting} if and only if $S_i$ is an accepting set.  A run $(T_r,r)$ is
\emph{accepting} if all its infinite paths are accepting.
% A 2WATA $\At$ \emph{accepts} a labeled tree $T$ from a node $x_0\in\Delta^T$
% if there exists an accepting run of $\At$ over $T$ from~$x_0$.
A node $x$ is \emph{selected} by a 2WATA $\At$ from a labeled tree $T$ if there
exists an accepting run of $\At$ over $T$ from $x$.
% The \emph{language} $\lang(\At)$ accepted by $\At$ is the set of trees that
% $\At$ accepts from the root $\varepsilon$.

\subsection{Binary Trees and Sibling Trees}

As mentioned before, we assume that \muXPath queries are expressed over binary
sibling trees, where the left child of a node
corresponds to the $\afchild$ axis, and the right
child corresponds to the $\aright$ axis.
To ensure that generic binary trees (i.e., trees of branching degree~2)
represent binary sibling trees, we make use of special propositions $\ifc$,
$\irs$, $\hfc$, $\hrs$.  The proposition $\ifc$ (resp., $\irs$) is used to keep
track of whether a node \emph{is the first child} (resp., \emph{is the right
 sibling}) of its parent, and $\hfc$ (resp.,
$\hrs$) is used to keep track of whether a node \emph{has a first child}
(resp., \emph{has a right sibling}).
In particular, we consider binary trees whose nodes are labeled with subsets of
$\Sigma_M=\Sigma\cup\{\ifc,\irs,\hfc,\hrs\}$.  We call such a tree
$T=\Tree[T]$ a \emph{well-formed binary tree} if it satisfies the following
conditions:
\begin{itemize}
\item For each node $x$ of $T$, if $\Lab{x}$ contains $\hfc$, then $x\conc 1$
  is meant to represent the $\afchild$ child of $x$
  and hence $\Lab{x\conc 1}$ contains $\ifc$ but not $\irs$.  Similarly, if
  $\Lab{x}$ contains $\hrs$, then $x\conc 2$ is meant to represent the
  $\aright$ child of $x$ and hence $\Lab{x\conc 2}$
  contains $\irs$ but not $\ifc$.
\item The label $\Lab{\eword}$ of the root of $T$ contains neither $\ifc$, nor
  $\irs$, nor $\hrs$.  In this way, we restrict the root of $T$ so as to
  represent the root of a sibling tree.
\end{itemize}
Notice that every (binary) sibling tree $T$ trivially induces a well-formed
binary tree $\pi_b(T)$ obtained by simply adding the labels $\ifc$, $\irs$,
$\hfc$, $\hrs$ in the appropriate nodes.

On the other hand, a well-formed binary tree $T=\Tree[T]$ induces a sibling
tree $\pi_s(T)$.  To define $\pi_s(T)=(\Dom[\Ts],\Int[\Ts]{\cdot})$, we define,
by induction on $\Dom$, a mapping $\pi_s$ from $\Dom$ to words over $\Nat$ as
follows:
\begin{itemize}
\item $\pi_s(\eword)=\eword$;
\item if $\hfc\in\Lab{x}$, then $\pi_s(x\conc 1)=\pi_s(x)\conc 1$;

\item if $\hrs\in\Lab{x}$ and $\pi_s(x)=z\conc n$, with $z\in\Nat^*$ and
  $n\in\Nat$, then $\pi_s(x\conc 2)=z\conc (n{+}1)$.
\end{itemize}
Then, we take $\Dom[\Ts]$ to be the range of $\pi_s$, and we define the
interpretation function $\Int[\Ts]{\cdot}$ as follows: for each $A\in\SigmaA$,
we define $\Int[\Ts]{A}=\{\pi_s(x)\in\Dom[\Ts]\mid A\in\Lab{x}\}$.
% similarly, for each $\Id\in\SigmaId$, we define
% $\Int[\Ts]{\Id}=\{\pi_s(x)\in\Dom[\Ts]\mid \Id\in\Lab{x}\}$. Notice that,
% since $\Ts$ is well-formed, $\Int[\Ts]{\Id}$ contains at most one element.
%
Note that the mapping $\pi_s$ ignores irrelevant parts of the binary tree,
e.g., if the label of a node $x$ does not contain $\hfc$, even if $x$ has a
$1$-child, such a node is not included in the sibling
tree.  Note also that we cannot rule out a priori the presence or
 irrelevant nodes, since our techniques for query satisfiability and
 containment rely on non-emptiness of well-formed binary trees (cf.\
 Section~\ref{sec:satisfiability-containment}).
%
% Also, $\pi_s$ can be considered the inverse of the mapping $\pi_b$ defined in
% Section~\ref{sec-evaluation}.

\subsection{From \muXPath to 2WATAs}
\label{sec:muXPath-to-2wata}

We show now how to construct \myi~from each \muXPath query $\varphi$ (over
binary sibling trees) a 2WATA $\At_{\varphi}$ whose number of states is linear
in $|\varphi|$ and that selects from a tree $T$ precisely the nodes in
$\Int[T]{\varphi}$, and \myii~from each 2WATA $\At$ a \muXPath query
$\varphi_{\At}$ % that is
of size linear in the number of states of $\At$ % and
that, when evaluated over a tree $T$, returns precisely the nodes selected by
$\At$ from $T$.

In order to translate \muXPath to 2WATAs, we need to make use of a notion of
syntactic closure, similar to that of Fisher-Ladner closure of a formula of PDL
\cite{FiLa79}.
The \emph{syntactic closure} $\CL(X:\F)$ of a \muXPath query $X:\F$ is defined
as $\{\ifc,\irs,\hfc,\hrs\}\cup\CL(\F)$, where $\CL(\F)$ is defined as follows:
for each equation $X\doteq\varphi$ in some fixpoint block in $\F$,
$\{X,\nnf(\varphi)\}\subseteq\CL(\F)$, where $\nnf(\psi)$ denotes the negation
normal form of $\psi$, and then we close the set under sub-expressions (in
negation normal form), by inductively applying the rules in
Figure~\ref{fig:closure}.
The \emph{negation normal form} of an XPath node expression is
 obtained in the standard way, by pushing negation inside operators as much as
 possible, i.e., by recursively replacing
\[
  \begin{array}{rcl}
    \lnot(\varphi_1\land\varphi_2) &\text{by}&
      \lnot\varphi_1\lor\lnot\varphi_2,\\
    \lnot(\varphi_1\lor\varphi_2) &\text{by}&
      \lnot\varphi_1\land\lnot\varphi_2,
  \end{array}
  \qquad\qquad
  \begin{array}{rcl}
    \lnot\DIAM{P}{\varphi} &\text{by}& \BOX{P}{\lnot\varphi},\\
    \lnot\BOX{P}{\varphi} &\text{by}& \DIAM{P}{\lnot\varphi}
  \end{array}
\]
until negation appears in front of atomic propositions only.
%
% \begin{proposition}\label{prop:cardinality-CL}
%   Given a \muXPath query $\query$, the cardinality of $\CL(\query)$ is
%   linear in the length of $\query$.
% \end{proposition}
It is easy to see that, for a \muXPath query $\query$, the cardinality of
$\CL(\query)$ is linear in the length of $\query$.

\begin{figure}[tbp]
  \[
   \begin{array}{@{\text{if}~}l@{~\text{then}~}l}
     \psi \in\CL(\varphi) &
       \nnf(\lnot\psi) \in\CL(\varphi),
       \qquad\text{if $\psi$ is not of the form $\lnot\psi'$}\\
     \lnot\psi \in\CL(\varphi) & \psi \in\CL(\varphi)\\
     \psi_1\land\psi_2 \in\CL(\varphi) & \psi_1,\, \psi_2 \in\CL(\varphi)\\
     \psi_1\lor \psi_2 \in\CL(\varphi) & \psi_1,\, \psi_2 \in\CL(\varphi)\\
     %% \DIAM{P}{\psi} \in\CL(\varphi) & \psi\in\CL(\varphi)\\
     \DIAM{P}{\psi} \in\CL(\varphi) & \psi\in\CL(\varphi),
       \qquad\text{for $P\in\{\afchild,\aright,\afchild^-,\aright^-\}$}\\
     \BOX{P}{\psi} \in\CL(\varphi) & \psi\in\CL(\varphi),
       \qquad\text{for $P\in\{\afchild,\aright,\afchild^-,\aright^-\}$}
   \end{array}
  \]
%\begin{small}
%  \centering
%  \begin{tabular}{@{}l}
%    If $\psi\in\CL(\varphi)$ then $\nnf(\lnot\psi)\in\CL(\varphi)$
%    \quad (if $\psi$ is not of the form $\lnot\psi'$).\\
%    If $\lnot\psi\in\CL(\varphi)$ then $\psi\in\CL(\varphi)$.\\
%    If $\psi_1\land\psi_2\in\CL(\varphi)$ or $\psi_1\lor\psi_2\in\CL(\varphi)$
%    then $\psi_1,\, \psi_2 \in\CL(\varphi)$.\\
%    If $\DIAM{P}{\psi}\in\CL(\varphi)$ or $\BOX{P}{\psi}\in\CL(\varphi)$ then
%    $\psi\in\CL(\varphi)$.
%  \end{tabular}
% \end{small}
  \caption{Closure of \muXPath expressions}
  \label{fig:closure}
\end{figure}

Let $\query=X_0:\F$ be a \muXPath query.  We show how to construct a 2WATA
$\At_{\query}$ that, when run over a well-formed binary tree $T$, accepts
exactly from the nodes in $\Int[T]{\query}$.
The 2WATA
$\At_{\query}=(\L,S_{\query},s_{\query},\delta_{\query},\alpha_{\query})$
is defined as follows.
\begin{itemize}
\item The alphabet is $\L=2^{\Sigma_M}$, with
   $\Sigma_M=\Sigma\cup\{\ifc,\irs,\hfc,\hrs\}$.
  This corresponds to labeling each node of the tree with a truth assignment to
  the atomic propositions, including the special ones that encode information
  about the parent node and about whether the
  children are significant.

\item The set of states is $S_{\query}=\CL(\query)$.
  Intuitively, when the automaton is in a state $\psi\in\CL(\query)$ and visits
  a node $x$ of the tree, it checks that the node expression $\psi$ holds in
  $x$.

\item The initial state is $s_{\query}=X_0$.

\item The transition function $\delta_{\query}$ is defined as follows:

\begin{enumerate}
\item \label{enu-transitions-atomic} For each $\lambda\in\L$, and each
  % symbol
  $\sigma\in\Sigma_M$,
  \[
    \begin{array}{rcl}
      \delta_{\query}(\sigma,\lambda) &=&
        \begin{cases}
          \true,  & \text{if $\sigma\in\lambda$}\\
          \false, & \text{if $\sigma\notin\lambda$}
        \end{cases}
      \\
      \delta_{\query}(\lnot\sigma,\lambda) &=&
        \begin{cases}
          \true,  & \text{if $\sigma\notin\lambda$}\\
          \false, & \text{if $\sigma\in\lambda$}
        \end{cases}
    \end{array}
  \]
  Such transitions check the truth value of atomic propositions, and of their
  negations in the current node of the tree, by simply checking whether the
  node label contains the proposition or not.

\item \label{enu-transitions-nonatomic} For each $\lambda\in\L$ and each formula
  $\psi\in\CL(\query)$, the automaton inductively decomposes $\psi$ and moves
  to appropriate states to check the sub-expressions as follows:
  \[
      \begin{array}[t]{@{}rcl}
        \delta_{\query}(\psi_1\land\psi_2,\lambda)
        &=& (0,\psi_1) ~\land~ (0,\psi_2)\\
        \delta_{\query}(\psi_1\lor\psi_2,\lambda)
        &=& (0,\psi_1) ~\lor~ (0,\psi_2)
        \\[2mm]
        \delta_{\query}(\DIAM{\afchild}{\psi},\lambda)
        &=& (0,\hfc) \land (1,\psi)\\
        \delta_{\query}(\DIAM{\aright}{\psi},\lambda)
        &=& (0,\hrs) \land (2,\psi)\\
        \delta_{\query}(\DIAM{\afchild^-}{\psi},\lambda)
        &=& (0,\ifc) \land (-1,\psi)\\
        \delta_{\query}(\DIAM{\aright^-}{\psi},\lambda)
        &=& (0,\irs) \land (-1,\psi)
        \\[2mm]
        \delta_{\query}(\BOX{\afchild}{\psi},\lambda)
        &=& (0,\lnot\hfc) \lor (1,\psi)\\
        \delta_{\query}(\BOX{\aright}{\psi},\lambda)
        &=& (0,\lnot\hrs) \lor (2,\psi)\\
        \delta_{\query}(\BOX{\afchild^-}{\psi},\lambda)
        &=& (0,\lnot\ifc) \lor (-1,\psi)\\
        \delta_{\query}(\BOX{\aright^-}{\psi},\lambda)
        &=& (0,\lnot\irs) \lor (-1,\psi)
      \end{array}
  \]

\item \label{enu-transitions-equation} Let $X\doteq\varphi$ be an equation in
  one of the blocks of $\F$.  Then, for each $\lambda\in\L$, we have
  $\delta_{\query}(X,\lambda) = (0,\varphi)$.
\end{enumerate}

\item To define the weakness partition of $\At_{\query}$, we partition the
  expressions in $\CL(\query)$ according to the partial order on the fixpoint
  blocks in $\F$.  Namely, we have one element of the partition for each
  fixpoint block $F\in\F$.  Such an element is formed by all expressions
  (including variables) in $\CL(\query)$ in which at least one variable defined
  in $F$ occurs and no variable defined in a fixpoint block $F'$ with $F\prec
  F'$ occurs.  In addition, there is one element of the partition consisting of
  all expressions in which no variable occurs.
  Then the acceptance condition $\alpha_{\query}$ is the union of all elements
  of the partition corresponding to a greatest fixpoint block.
  Observe that the partial order on the fixpoint blocks in $\F$
  % , due to the alternation freedom of $\query$,
  guarantees that the transitions of $\At_{\query}$ satisfy the weakness
  condition.  In particular, each element of the weakness partition is either
  contained in $\alpha_{\query}$ or disjoint from $\alpha_{\query}$.
  This guarantees that an accepting run cannot get trapped in a state
  corresponding to a least fixpoint block, while it is allowed to stay forever
  in a state corresponding to a greatest fixpoint block.
\end{itemize}

% Indeed, for least fixpoint blocks the run must terminate, which means that
% the least fixpoint converges.  Instead, for greatest fixpoint blocks the run
% is allowed to continue forever.

\begin{theorem}
  \label{thm:muXPath-to-2WATA}
  Let $\query$ be a \muXPath query.  Then:
  \begin{enumerate}
  \item The number of states of the corresponding 2WATA $\At_{\query}$ is
    linear in the size of $\query$.
  \item For every binary sibling tree $T$, a node $x$ of $T$ is in
    $\Int[T]{\query}$ iff $\At_{\query}$ selects $x$ from the well-formed
    binary tree $\pi_b(T)$ induced by $T$.\footnote{Observe that the trees $T$
     and $\pi_b(T)$ have the same nodes, and they differ only in the label of
     nodes.}
  \end{enumerate}
\end{theorem}

\begin{proof}
Item~(1) follows immediately from the fact that the size of $\CL(\query)$ is
linear in the size of $\query$.
We turn to item~(2).
% We use a construction inspired by~\cite{KuVW00}.
In the proof, we blur the distinction between $T$ and $\pi_b(T)$, denoting it
simply as $T$, since the two trees are identical, except for the additional
labels in $\pi_b(T)$, which are considered by $\At_{\query}$ but ignored by
$\query$.

Let $\query=X:\F$.  We show by simultaneous induction on the structure of $\F$
and on the nesting of fixpoint blocks, that for every expression
$\psi\in\CL(\F)$ and for every node $x$ of $T$, we have that $\At_{\query}$,
when started in state $\psi$, selects $x$ from $T$ if and only if
$x\in\Int[T]{\psi}$.
\begin{itemize}
\item Indeed, when $\psi$ is an atomic proposition, then the claim follows
  immediately by making use of the transitions in item~(1) of the definition of
  $\delta$.
\item When $\psi=\psi_1\land\psi_2$ or $\psi=\psi_1\lor\psi_2$, the claim
  follows by inductive hypothesis, making use of the first two transitions in
  item~(2).
\item When $\psi=\DIAM{\afchild}{\psi_1}$, the 2WATA checks that $x$ has a
  first child $y=x\cdot 1$, and moves to $y$ checking that $y$ is selected from
  $T$ starting in state $\psi_1$.  By induction
  hypothesis, this holds iff
  $y\in\Int[T]{\psi_1}$, and the claim follows.

  The cases of $\psi=\DIAM{\aright}{\psi_1}$, $\psi=\DIAM{\afchild^-}{\psi_1}$,
  and $\psi=\DIAM{\aright^-}{\psi_1}$ are analogous.
\item When $\psi=\BOX{\afchild}{\psi_1}$, the 2WATA checks that either $x$ does
  not have a first child, or that the first child $y=x\cdot 1$ is selected from
  $T$ starting in state $\psi_1$.  By induction
  hypothesis, this holds iff
  $y\in\Int[T]{\psi_1}$, and the claim follows.

  The cases of $\psi=\BOX{\aright}{\psi_1}$, $\psi=\BOX{\afchild^-}{\psi_1}$,
  and $\psi=\BOX{\aright^-}{\psi_1}$ are analogous.

\item When $\psi=Z$, let $Z\doteq\varphi$ be the equation defining $Z$.
  Then according to the transitions in item~(3), the 2WATA checks that $x$ is
  selected from $T$ starting in state $\varphi$.
  The definition of the 2WATA acceptance condition $\alpha_{\query}$ guarantees
  that, if $Z$ is defined in a least fixpoint block then an accepting run
  cannot get trapped in the element $S_i$ of the weakness partition containing
  $Z$; instead, if $Z$ is defined in a greatest fixpoint block then an
  accepting run is allowed to stay forever in $S_i$.

  We consider only the least fixpoint case; the greatest fixpoint case is
  similar.
  If there is an accepting run, it will go through states in $S_i$ (including
  $Z$) only a finite number of times, and on each of its paths it will get to
  a node $y$ in a state $\xi\in S_j$, where $S_j$ strictly precedes $S_i$,
  i.e., with $S_j\leq S_i$ and $S_j\neq S_i$.  By induction on the nesting of
  fixpoint blocks (each level of nesting corresponds to an
   element of the state partition), we have
  that $\At_{\query}$, when started in state $\xi$, selects $y$ from $T$ if and
  only if $y\in\Int[T]{\xi}$.  Then, since the automaton state $Z$ is not
  contained in $\alpha_{\query}$, the acceptance condition ensures that the
  transition in item~(3) is applied only a finite number of times, and
  considering the least fixpoint semantics, by structural induction we get that
  $x\in\Int[T]{Z}$.

  For the other direction, we show that, if $x\in\Int[T]{Z}$, then
  $\At_{\query}$ has an accepting run $R=\Tree[R]$ witnessing that $x$ is
  selected from $T$ starting in state $Z$.  We can define the run $R$ by
  exploiting the equation $Z\doteq\varphi$ to make the transition according to
  item~(3), and following structural induction to decompose formulas, ensuring
  that, for all nodes $y\in\Dom[R]$ with $\Lab[R]{y}=(x',\psi')$ we have that
  $x'\in\Int[T]{\psi'}$.  In particular, we need to resolve the
  nondeterminism coming from disjunctions in the transition function of
  $\At_{\query}$ (in turn coming from disjunctions in
  $\query$). Intuitively, we do so by choosing the disjunct that is
  satisfied in the node of $T$, ensuring that we do not postpone
   indefinitely the satisfaction of the fixpoint.  More specifically, since
   $x\in\Int[T]{Z}$ is defined by a least-fixed point (over its block), we have
   that $x\in\Int[T]{Z_i}$, for some approximant $Z_i$ of $Z$. Such $Z_i$ is
   obtained by recursively unfolding the definitions of $Z$ and of the other
   variables in its block, at most a finite number $i$ of times.  Hence, we can
   choose the disjuncts in the corresponding run $R$ in a way that the variable
   $Z$, and the other variables in its block, are reintroduced at most $i$
   times each, before moving to the next partition.  So, the run $R$ does not
   loop on the partition containing $Z$, and does not violate the acceptance
   condition.
\end{itemize}
The claim then follows since the initial state of $\At_{\query}$ is $\query$.
\end{proof}

We observe that, although the number of states of $\At_{\query}$ is linear in
the size of $\query$, the alphabet of $\At_{\query}$ is the powerset of that of
$\query$, and hence the transition function and the entire $\At_{\query}$ is
exponential in the size of $\query$.  However, as we will show later, this does
not affect the complexity of query evaluation, query containment, and more in
general reasoning over queries.\footnote{Observe that we could also change the
 transition function of 2WATAs, by making the dependency on the node label
 implicit.  Specifically, we could replace $\delta(s,\lambda)=\varphi$ with
 $\delta(s)=\varphi'$, and add to $\varphi'$ tests that check whether a symbol
 $\sigma$ is in the node label $\lambda$.  This would allows us to define a
 linear translation from \muXPath queries into automata, at the cost of using a
 non-standard definition of 2WATAs.}

\subsection{From 2WATAs to \muXPath}

We show now how to convert 2WATAs into \muXPath queries while preserving the
set of nodes selected from (well formed) binary trees.

Consider a 2WATA $\At=(\L,S,s_0,\delta,\alpha)$, where $\L=2^{\Sigma_M}$,
with $\Sigma_M= \Sigma\cup\{\ifc,\irs,\hfc,\hrs\}$, and let
$S=\cup_{i=1}^k S_i$ be the weakness partition of $\At$.
We define a translation $\pi$ as follows.
\begin{itemize}
\item For a positive Boolean formula $f\in\B^+([-1..2]\times S)$, we define a
  \muXPath node expression $\pi(f)$ inductively as follows:
  \[
    \begin{array}{@{}r@{~}c@{~}l@{\qquad\qquad}r@{~}c@{~}l}
      \pi(\false) &=& \FALSE
      & \pi(\true) &=& \TRUE\\
      \pi((1,s)) &=& \DIAM{\afchild}{s}
      & \pi((2,s)) &=& \DIAM{\aright}{s}\\
      \pi((0,s)) &=& s
      & \pi((-1,s)) &=&
      \DIAM{\afchild^-}{s} \lor \DIAM{\aright^-}{s}\\
      \pi(f_1\land f_2) &=& \pi(f_1)\land\pi(f_2)
      & \pi(f_1\lor f_2) &=& \pi(f_1)\lor\pi(f_2)\\
    \end{array}
  \]

\item For each state $s\in S$, we define a \muXPath equation $\pi(s)$ as
  follows:
  \[
    \textstyle
    s ~\doteq~ \bigvee_{\lambda\in\L}
    (\tilde{\lambda}\land \pi(\delta(s,\lambda))),
  \]
  where $\tilde{\lambda}=(\bigwedge_{a\in\lambda}a) \land
  (\bigwedge_{a\in(\Sigma_M\setminus\lambda)}\lnot a)$.
\item For each element $S_i$ of the weakness partition, we define a \muXPath
  fixpoint block as follows:
  \[
    \pi(S_i) ~=~
    \begin{cases}
      \GFP{\pi(s)\mid s\in S_i}, & \text{if } S_i\subseteq \alpha\\
      \LFP{\pi(s)\mid s\in S_i}, & \text{if } S_i\cap\alpha=\emptyset
    \end{cases}
  \]

\item Finally, we define the \muXPath query $\pi(\At)$ as:
  \[
    \pi(\At) ~=~ s_0:\{\pi(S_1),\dots,\pi(S_k)\}.
  \]
\end{itemize}

\begin{theorem}
  \label{thm:2WATA-to-muXPath}
  Let $\At$ be a 2WATA. Then:
  \begin{enumerate}
  \item The length of the \muXPath query $\pi(\At)$ is linear in the size of $\At$.
  \item For every binary sibling tree $T$, we have that $\At$ selects a node
    $x$ from $\pi_b(T)$ iff $x$ is in $\INT[T]{\pi(\At)}$.
  \end{enumerate}
\end{theorem}

\begin{proof}
Item~1 follows immediately from the above construction.  We only observe that
in defining the \muXPath equations $\pi(s)$ for a state $s\in S$, we have a
disjunction over the label set $\L$, which is exponential in the number of
atomic propositions in $\Sigma$.  On the other hand, the transition function of
the 2WATA itself needs to deal with the elements of $\L$, and hence is also
exponential in the size of $\Sigma$.

We turn to item~2.  We again ignore the distinction between $T$ and $\pi_b(T)$.
It is easy to check that, by applying the construction in
Section~\ref{sec:muXPath-to-2wata} to the \muXPath query $\pi(\At)$, we obtain
a 2WATA $\At_{\pi(\At)}$ that on well-formed binary trees is equivalent to
$\At$.
Indeed, for every transition of $\At$, the construction introduces in
$\At_{\pi(\At)}$ corresponding transitions that mimick it.
In particular, for an atom of the form $(-1,s)$ appearing in the right-hand
side of a transition of $\At$, we obtain in $\pi(\At)$ a \muXPath expression
$\varphi=\DIAM{\afchild^-}{s} \lor \DIAM{\aright^-}{s}$.  Then we have that
$\delta_{\pi(\At)}(\varphi,\lambda) = (0,\DIAM{\afchild^-}{s}) \lor
(0,\DIAM{\aright^-}{s})$, and for such resulting states we have in turn
transitions
$\delta_{\pi(\At)}(\DIAM{\afchild^-}{s},\lambda)=(0,\ifc)\land (-1,s)$ and
$\delta_{\pi(\At)}(\DIAM{\aright^-}{s},\lambda)=(0,\irs)\land (-1,s)$.  In both
cases where $\lambda$ contains $\ifc$ or $\irs$, this expression results in
$(-1,s)$, while in the root (where $\lambda$ contains neither $\ifc$ nor
$\irs$) the expression results in $\false$, thus yielding a transition
equivalent to the one resulting from the atom $(-1,s)$ of $\At$.
Hence, by Theorem~\ref{thm:muXPath-to-2WATA}, we get the claim.
\end{proof}

%%% Local Variables:
%%% mode: latex
%%% TeX-master: "main"
%%% End:

\section{Acceptance and Non-Emptiness for 2WATAs}
\label{sec:2wata}

We provide now computationally optimal algorithms for deciding the acceptance
and non-emptiness problems for 2WATAs.

\subsection{The Acceptance Problem}

Given a 2WATA $\At=(\L,S,s_0,\delta,\alpha)$, a labeled tree $T=\Tree[T]$, and
a node $x_0\in\Delta^T$, we would like to know whether $x_0$ is selected by
$\At$ from $T$.  This is called the \emph{acceptance problem}.  We follow here
the approach of \cite{KuVW00}, and solve the acceptance problem by first taking
a product $\At\times T_{x_0}$ of $\At$ and $T$ from~$x_0$. This product is an
alternating automaton over a one letter alphabet $\L_0$, consisting of a single
letter, say $a$. This product automaton simulates a run of $\At$ on $T$
from~$x_0$.  The product automaton is
$\At\times T_{x_0} =
(\L_0,S\times\Delta^T,(s_0,x_0),\delta',\alpha\times\Delta^T)$, where $\delta'$
is defined as follows:
\begin{itemize}
\item
$\delta'((s,x),a)=\Theta_x(\delta(s,\ell^T(x)))$, where $\Theta_x$ is the
substitution that replaces a pair $(c,t)$ in $\delta(s,\ell^T(x))$
by the pair $(t,x\conc c)$ if $x\conc c \in\Delta^T$, and by \false
otherwise.
\end{itemize}
Note that the size of $\At\times T_{x_0}$ is simply the product of the size of
$\At$ and the size of $T$, and that the only elements of $\L$ that are used in
the construction of $\At\times T_{x_0}$ are those that appear among the labels
of $T$.  Note also that $\At\times T_{x_0}$ can be viewed as a weak alternating
word automaton running over the infinite word $a^\omega$, as by taking the
product with $T$ we have eliminated all directions.  In fact, one can simply
view $\At\times T_{x_0}$ as a 2-player infinite game; see \cite{GrTW02}.

We can now state the relationship between $\At\times T_{x_0}$ and $\At$,
which is essentially a restatement of Proposition~3.2 in \cite{KuVW00}.
\begin{proposition}\label{prop:one-letter}
  Node $x_0$ is selected by $\At$ from $T$ iff $\At\times T_{x_0}$ accepts
  $a^\omega$.
\end{proposition}

The advantage of Proposition~\ref{prop:one-letter} is that it reduces the
acceptance problem to the question of whether $\At\times T_{x_0}$ accepts
$a^\omega$. This problem is referred to in \cite{KuVW00} as the ``one-letter
nonemptiness problem''. It is shown there that this problem can be solved in
time that is linear in the size of $\At\times T_{x_0}$ by an algorithm that
imposes an evaluation of and-or trees over a decomposition of the automaton
state space into maximal strongly connected components, and then analyzes these
strongly connected components in a bottom-up fashion.  The result
in \cite{KuVW00} is actually stronger; the algorithm there computes in linear
time the set of states from which the automaton accepts $a^\omega$, that is,
the states that yield acceptance if chosen as initial states.
We therefore obtain the following result about the acceptance problem.

\begin{theorem}\label{thm:2WATA-query-evaluation}
  Given a 2WATA $\At$ and a labeled tree $T$, we can compute the set of nodes
  selected by $\At$ from $T$ in time that is linear in the product of the sizes
  of $\At$ and $T$.
\end{theorem}

\begin{proof}
We constructed above the product automaton $\At\times T_{x_0} =
(\L_0,S\times\Delta^T,(s_0,x_0),\delta',\alpha\times\Delta^T)$.
Note that the only place in this automaton where $x_0$ plays a role
is in the initial state $(s_0,x_0)$. That is, replacing the initial
state by $(s_0,x)$ for another node $x\in\Delta^T$ gives us the product
automaton $\At\times T_{x}$.  As pointed out above, the bottom-up algorithm
of \cite{KuVW00} actually computes the set of states from which the automaton
accepts $a^\omega$. Thus, $x$ is selected by $\At$ from $T$ iff the state
$(s_0,x)$ of the product automaton is accepting. That is, to compute the
set of nodes of $T$ selected by $\At$, we construct the product automaton,
compute states from which the automaton accepts, and then select all nodes
$x$ such that the automaton accepts from $(s_0,x)$.
\end{proof}
Thus, Theorem~\ref{thm:2WATA-query-evaluation} provides us with a
query-evaluation algorithms for 2WATA queries, which is linear both in the size
of the tree and in the size of the automaton.

\subsection{The Nonemptiness Problem}
\label{sec:2WATA-to-NTA}

The \emph{nonemptiness problem} for 2WATAs consists in determining, for a
given 2WATA $\At$ whether it accepts some tree $T$ from $\eword$.
This problem is solved in~\cite{Vard98} for 2WATAs (actually, for
a more powerful automata model) over infinite trees, using rather
sophisticated automata-theoretic techniques. Here we solve this problem
over finite trees, which requires less sophisticated techniques,
and, consequently, is much easier to implement.

In order to decide non-emptiness of 2WATAs, we resort to a conversion to
standard one-way nondeterministic tree automata~\cite{CDGJ*02}.
A one-way nondeterministic tree automaton (NTA) is a tuple
$\At=(\L, S, s_0,\delta)$, analogous to a 2WATA, except that
\myi~the acceptance condition $\alpha$ is empty and has been dropped from
the tuple, \myii~the directions $-1$ and $0$ are not used in $\delta$ and,
\myiii~for each state $s\in S$ and letter $a\in\L$, the positive Boolean
formula $\delta(s,a)$, when written in DNF, does not contain a disjunct
with two distinct atoms $(c,s_1)$ and $(c,s_2)$ with the same direction~$c$.
In other words, each disjunct corresponds to sending at most one
``subprocess'' in each direction. We also allow an NTA to have a \emph{set}
of initial states, requiring that starting with \emph{one} initial state
must lead to acceptance.

While for 2WATAs we have separate input
tree and run tree, for NTAs we can assume that the run of the automaton
over an input tree $T=\Tree[T]$ is an $S$-labeled tree
$R=(\Dom[T],\lab[R])$, which has the same underlying tree
as $T$, and thus is finite, but is labeled by states in~$S$.
Nonemptiness of NTAs is known to be decidable~\cite{Done65}.
As shown there, the set $Acc$ of states of an NTA that leads to acceptance
can be computed by a simple fixpoint algorithm:
\begin{enumerate}[(1)]
\item Initially: $Acc=\emptyset$.
\item At each iteration: $Acc:= Acc \cup\{s \mid \alpha_{Acc} \models
   \delta(s,a)\text{ for some }a\in\L\}$, where $\alpha_X$ is the truth
  assignment that maps $(c,s)$ to true precisely when $s\in X$,
\end{enumerate}
It is known that such an algorithm can be implemented to run in linear time
\cite{DoGa84}.  Thus, to check nonemptiness we compute $Acc$ and check that it
has nonempty intersection with the set of initial states.

It remains to describe the translation of 2WATAs to NTAs.  Given a 2WATA $\At$
and an input tree $T$ of branching degree $k$, let
$\Tau= 2^{S \times [-1..k] \times S}$; that is, an element of $\Tau$ is a set
of transitions of the form $(s,i,s')$.  A \emph{strategy for $\At$ on} $T$ is a
mapping $\tau : \Delta^T \rightarrow \Tau$. Thus, each label in a strategy is
an edge-$[-1..k]$-labeled directed graph on $S$.  For each label
$\zeta\subseteq S \times [-1..k] \times S$, we define
$\sta(\zeta)=\{u \mid (u,i,v)\in \zeta\}$, i.e., $\sta(\zeta)$ is the set
of sources in the graph $\zeta$.  In addition, we require the following:
\begin{enumerate}[(1)]
\item
for each node $x\in\Delta^T$ and each state $s\in\sta(\tau(x))$,
the set $\{(c,s') \mid (s,c,s')\in\tau(x)\}$ satisfies $\delta(s,\ell^T(x))$
(thus, each label can be viewed as a strategy of satisfying the transition
function), and
\item
for each node $x\in\Delta^T$, and each edge $(s,i,s')\in \tau(x)$,
we have that $s'\in\sta(\tau(x\conc i))$.
\end{enumerate}

A \emph{path} $\beta$ in the strategy $\tau$ is a maximal sequence
$(u_0,s_0),(u_1,s_1),\ldots$ of pairs from $\Delta^T \times S$
such that $u_0=\varepsilon$ and, for all $i \geq 0$, there is
some $c_i\in [-1..k]$ such that $(s_i,c_i,s_{i+1})\in\tau(u_i)$ and
$u_{i+1} = u_i \conc c_i$. Thus, $\beta$ is obtained by following
transitions in the strategy. The path $\beta$ is accepting if the
path $s_0,s_1,\ldots$ is accepting.
The strategy $\tau$ is \emph{accepting} if all its paths are accepting.

\begin{proposition}[\cite{Vard98}] \label{tree1}
  A 2WATA $\At$ accepts an input tree $T$ from $\varepsilon$
  iff $\At$ has an accepting strategy for $T$.
\end{proposition}

We have thus succeeded in defining a notion of run for alternating
automata that will have the same tree structure as the input tree.
We are still facing the problem that paths in a strategy tree can
go both up and down. We need to find a way to restrict attention
to uni-directional paths. For this we need an additional concept.

Let $\Eta$ be the set of relations of the form $S \times \{0,1\}\times S$.
Thus, each element in $\Eta$ is an edge-$\{0,1\}$-labeled
directed graph on $S$.
An \emph{annotation} for $\At$ on $T$ with respect to a strategy $\tau$
is a mapping $\eta : \Delta^T \rightarrow 2^{S \times \{0,1\}\times S}$.
Edge labels need not be unique; that is,
an annotation can contain both triples $(s,0,s')$ and $(s,1,s')$.
We require $\eta$ to satisfy some  closure conditions for each
node $x\in \Delta^T$. Intuitively, these conditions say that
$\eta$ contains all relevant information about finite paths in
$\tau$. Thus, an edge $(s,c,s')$ describes a path from $s$ to $s'$,
where $c=1$ if this path goes through~$\alpha$.  The conditions are:
\begin{enumerate}
\item
if $(s,c,s')\in \eta(x)$ and $(s',c',s'')\in\eta(x)$, then
$(s,c'',s'')\in\eta(x)$ where $c''=\max\{c,c'\}$,
\item
if $(s,0,s')\in \tau(x)$ then $(s,c,s')\in\eta(x)$,
where $c=1$ if $s'\in\alpha$ and $c=0$ otherwise,
\item
if $y=x\conc i$, $(s,i,s')\in\tau(x)$, $(s',c,s'')\in\eta(y)$, and
$(s'',-1,s''')\in\tau(y)$, then $(s,c',s''')\in\eta(x)$,
where $c'=1$ if $s\in\alpha$, $c=1$, or $s'''\in\alpha$,
and $c'=0$ otherwise.
\item
if $x=y\conc i$, $(s,-1,s')\in\tau(x)$, $(s',c,s'')\in\eta(y)$, and
$(s'',i,s''')\in\tau(y)$, then $(s,c',s''')\in\eta(x)$,
where $c'=1$ if either $s'\in\alpha$, $c=1$, or $s'''\in\alpha$,
and $c'=0$ otherwise.
\end{enumerate}
The annotation $\eta$ is \emph{accepting} if for every node $x\in \Delta^T$ and
state $s\in S$, if $(s,c,s)\in\eta(x)$, then $c=1$.
In other words, $\eta$ is accepting if \emph{all} cycles visit accepting states.

\begin{proposition}[\cite{Vard98}] \label{tree2}
A 2WATA $\At$ accepts an input tree $T$ from $\varepsilon$
iff $\At$ has a anchored strategy $\tau$
on $T$ and an accepting annotation $\eta$ of $\tau$.
\end{proposition}

Consider now an \emph{annotated tree} $(\Dom[T],\lab[T],\tau,\eta)$,
where $\tau$ is a strategy tree for $\At$ on $(\Dom[T],\lab[T])$ and
$\eta$ is an annotation of $\tau$. We say that
$(\Dom[T],\lab[T],\tau,\eta)$ is \emph{accepting} if $\eta$ is accepting.

\begin{theorem}\label{thm-annotated-trees}
Let $\At$ be a 2WATA. Then there is an NTA $\At^n$ such that
$\lang(\At)=\lang(\At^n)$.  The number of states of $\At_n$ is at most
exponential in the number of states of $\At$.
\end{theorem}

\begin{proof}
The proof follows by specializing the construction in \cite{Vard98} to
2WATAs on finite trees.

Let $\At=(\L,S,s_0,\delta,\alpha)$ and let the input tree be $T=\Tree[T]$.  The
automaton $\At^n$ guesses mappings $\tau: \Delta^T \rightarrow \Tau$
and $\eta: \Delta^T \rightarrow \Eta$ and checks that $\tau$ is a strategy for
$\At$ on $T$ and $\eta$ is an accepting annotation for $\At$ on $T$ with respect
to $\tau$. The state space of $\At_n$ is $\Tau \times \Eta$;
intuitively, before reading the label of a node $x$, $\At_n$ needs to be in
state $(\tau(x),\eta(x)$. The transition function of $\At_n$ checks that
$state$, $\tau$, and $\eta$ satisfies all the required conditions.

Formally, $\At^n=(\L,Q,Q_0,\rho)$, where
\begin{itemize}
\item
$Q=\Tau \times \Eta$,
\item
We first define transitions
$\rho(P,r,a,i) \rightarrow \Tau\times\Eta$:

We have that $(r',R')\in\rho((P,r,R),a,i$ if
\begin{enumerate}
\item
if $(s,i,s)\in r$, then $s'\in state(r')$,
\item
if $(s,-1,s)\in r'$, then $s'\in state(r)$,
\item
if $(s,c,s)\in R$, then $c=1$,
\item
for each $s\in state(r)$,
the set $\{(c,s'): (s,c,s)\in r\}$ satisfies $\delta(s,a)$,
\item
if $(s,c,s')\in R'$ and $(s',c',s'')\in R'$, then
$(s,c'',s'')\in R'$ where $c''=\max\{c,c'\}$,
\item
if $(s,0,s')\in r$, then $(s,c,s')\in R$,
where $c=1$ if $s'\in\alpha$ and $c=0$ otherwise,
\item
if $(s,i,s')\in r$, $(s',c,s'')\in R'$, and
$(s'',-1,s''')\in r'$, then $(s,c',s''')\in r$,
where $c'=1$ if $s\in\alpha$, $c=1$, or $s'''\in\alpha$,
and $c'=0$ otherwise.
\item
if $(s,-1,s')\in r'$, $(s',c,s'')\in R$, and
$(s'',i,s''')\in r$, then $(s,c',s''')\in R'$,
where $c'=1$ if either $s'\in\alpha$, $c=1$, or $s'''\in\alpha$,
and $c'=0$ otherwise.
\end{enumerate}
Intuitively, the transition function $\rho$ checks that all conditions
on the strategy and annotation hold, except for the condition on the
strategy at the root.

We now define
$\rho(R,r,a)=\bigvee_{1\leq i \leq k} \bigwedge{(R',r')\in\rho(R,r,a,i)} (R',r'.i)$.
If, however, we have that
$(\emptyset,\emptyset,\emptyset)\in\rho((P,R,r),a,i$ for all $1\leq i \leq k$, then we define
$\rho(R,r,a)={\bf true}$.
\item
The set of initial states is
$Q_0=\{(r,R): s_0\in state(r) \mbox{ and there is no transition} (s,-1,s')\in r\}$.
\end{itemize}
It follows from the argument in \cite{Vard98} that $\At^n$ accepts a tree $T$
iff $\At$ has a strategy tree on $T$ and an accepting annotation of that
strategy.
\end{proof}

We saw earlier that nonemptiness of NTAs can be checked in linear time.
From Proposition~\ref{tree2} we now get:

\begin{theorem}
\label{thm-2WATA-nonemptiness}
Given a 2WATA $\At$ with $n$ states and an input alphabet with $m$~elements,
deciding nonemptiness of $\At$ can be done in time exponential in $n$ and
linear in $m$.
\end{theorem}

The key feature of the state space of $\At_n$ is the fact that
states are triples consisting of subsets of $S$, $S\times\{0,1\}\times S$,
and $S \times [-1..k] \times S$.
Thus, a set of states of $\At_n$ can be described by a Boolean function on
the domain $S^5 \times [-1..k]$. Similarly, the transition function of $\At_n$
can also be described as a Boolean function. Such functions can be
represented by binary decision diagrams (BDDs)~\cite{Brya86}, enabling a
symbolic implementation of the fixpoint algorithm discussed above.

We note that the framework of~\cite{Vard98} also converts a two-way alternating
tree automaton (on infinite trees) to a nondeterministic tree automaton (on infinite
trees).  The state space of the latter, however, is considerably more complex than
the one obtained her.In fact, the infinite-tree automata-theoretic approach
has so far have resisted attempts at practically efficient
implementation~\cite{ScTW05,TaHB95}, due to the use of Safra's determinization
construction~\cite{Safr88} and parity games~\cite{Jurd00}.
This makes it very difficult in practice to apply the symbolic approach in the
infinite-tree setting.

As shown in~\cite{Marx04b}, reasoning over \RXPath formulas can be reduced to
checking satisfiability in \emph{Propositional Dynamic Logics (PDLs)}.  Specifically,
one can resort to \emph{Repeat-Converse-Deterministic PDL (repeat-CDPDL)}, a variant
of PDL that allows for expressing the finiteness of trees and for which satisfiability
is \EXPTIME-complete~\cite{Vard98}.  This upper bound, however, is established using
sophisticated infinite-tree automata-theoretic techniques, which, we just point
out have resisted practically efficient implementations.  The main advantage of our
approach here is that we use only automata on finite trees, which require a much
``lighter'' automata-theoretic machinery. Indeed, symbolic-reasoning-techniques,
including BDDs and Boolean saisfiability solving have been used successfully
for XML reasoning \cite{GeLa10}.  We leave further exploration of this aspect
to future work.

%%% Local Variables:
%%% mode: latex
%%% TeX-master: "main"
%%% End:

%% \section{Query evaluation, satisfiability, containment, and view-based
%%  answering for \RXPath queries}

\section{Query Evaluation and Reasoning on \muXPath}
\label{sec:reasoning}

We now exploit the correspondence between \muXPath and 2WATAs we establish the
main characteristics of \muXPath as a query language over sibling trees.  We
recall that sibling trees can be encoded in (well-formed) binary trees in
linear time, and hence we blur the distinction between the two.

\subsection{Query Evaluation}

We can evaluate \muXPath queries over sibling trees by exploiting the
correspondence with 2WATAs, obtaining the following complexity
characterization.

\begin{theorem}\label{thm:muXPath-evaluation}
  Given a (binary) sibling tree $T$ and a \muXPath query $\query$, we can
  compute $\Int[T]{\query}$ in time that is linear in the number of nodes of
  $T$ (data complexity) and in the size of $\query$ (query complexity).
\end{theorem}

\begin{proof}
By Theorem~\ref{thm:muXPath-to-2WATA}, we can construct from $\query$ a 2WATA
$\At_{\query}$ whose number of states is linear in the size of $\query$.  On
the other hand, the well-formed binary tree $\pi_b(T)$ induced by $T$ can be
built in linear time.  By Theorem~\ref{thm:2WATA-query-evaluation}, we can
evaluate $\At_{\query}$ over $\pi_b(T)$ in linear time in the product of the
sizes of $\At_{\query}$ and $\pi_b(T)$ by constructing the product automaton
$\At_{\query}\times\pi_b(T)_x$ (where $x$ is an arbitrary node of $\pi_b(T)$).
Notice that, while the alphabet of $\At_{\query}$ is the powerset of the
alphabet of $\query$, in $\At_{\query}\times\pi_b(T)_x$ only the labels that
actually appear in $\pi_b(T)$ are used, hence the claim follows.
\end{proof}

\subsection{Query Satisfiability and Containment}
\label{sec:satisfiability-containment}

We now turn our attention to query satisfiability and containment.
A \muXPath query $\query$ is \emph{satisfiable} if there is a sibling tree
$\Ts$ and a node $x$ in $\Ts$ that is returned when $\query$ is evaluated over
$\Ts$.
A \muXPath query $\query_1$ is \emph{contained} in a \muXPath query $\query_2$
if for every sibling tree $\Ts$, the query $\query_1$ selects a subset of the
nodes of $\Ts$ selected by $\query_2$.
Checking satisfiability and containment of queries is crucial in several
contexts, such as query optimization, query reformulation, knowledge-base
verification, information integration, integrity checking, and cooperative
answering~\cite{GuUl92,BDHS96,Motr96,CKPS95,LeSa95,LeRo97b,MiSu99}.
% \cite{ACPS96,BDHS96,CDLNR98,CKPS95,FFLS99,GuUl92,LeRo97b,LeSa95,MiSu99,Motr96}.
Obviously, query containment is also useful for checking equivalence of
queries, i.e., verifying whether for all databases the answer to a query is the
same as the answer to another query.  For a summary of results on query
containment in graph and tree-structured data,
see~\cite{CDLV02b,BjMS11,Barc13}.

Satisfiability of a \muXPath query $\query$ can be checked by checking the
non-emptiness of the 2WATA $\At_{\query}$, after it has been extended in such a
way that it also checks that the accepted binary trees are well-formed (and
hence correspond to binary sibling trees).

Formally, we extend
$\At_{\query}=(\L,S_{\query},s_{\query},\delta_{\query},\alpha_{\query})$ to
the 2WATA $\Atwf_{\query}=(\L,S,\ini,\delta,\alpha)$ defined as follows:
\begin{itemize}
\item The set of states is $S=S_\query\cup \{\ini,\str\}$, where $\ini$ is the
  initial state, and $\str$ is used to check structural properties of
  well-formed trees.
\item The transition function is constituted by all transitions in
  $\delta_{\query}$, plus the following transitions ensuring that
  $\Atwf_{\query}$ accepts only well-formed trees.

\begin{enumerate}
\item \label{enu:transition-initial} For each $\lambda\in\L$, there is a
  transition
  \[
    \delta(\ini,\lambda) ~=~
    \begin{array}[t]{@{}l}
      (0,s_\query) \land
      (0,\str)
      % \land \bigwedge_{\Id\in\SigmaId} (0,\sid)
    \end{array}
  \]
  Such transitions both move to the initial state of $\At_{\query}$ to verify
  that $\query$ holds at the starting node, and move to state $\str$, from
  which structural properties of the tree are verified.
  %
  % move to states $\sid$ (for each $\Id\in\SigmaId$), from which the
  % automaton verifies that a single occurrence of $\Id$ is present on the
  % tree.

\item \label{enu:transition-structural} For each $\lambda\in\L$, there is a
  transition
  \[
    \delta(\str,\lambda) ~=~
    \begin{array}[t]{@{}l}
      ((0,\lnot\hfc) \lor ((1,\ifc) \land (1,\lnot\irs) \land (1,\str)))
      \land{}\\
      ((0,\lnot\hrs) \lor ((2,\irs) \land (2,\lnot\ifc) \land (2,\str)))
      \land{}\\
      (((0,\lnot\ifc) \land (0,\lnot\irs)) \lor (-1,\str))
    \end{array}
  \]
  Such transitions check that,
  %\begin{inparaenum}[\it (i)]
  %\item
  \textit{(i)}~for a node labeled with $\hfc$, its left child is labeled with
  $\ifc$ but not with $\irs$, and satisfies the same structural property;
  % \item
  \textit{(ii)}~for a node labeled with $\hrs$, its right child is labeled with
  $\irs$ but not with $\ifc$, and satisfies the same structural property; and
  % \item
  \textit{(iii)}~ for a node that is not the root, the predecessor satisfies
  the same structural property.
  %\end{inparaenum}

% \item \label{enu-transition-id} For each $\lambda\in\L$
%   and each $\Id\in\SigmaId$ there are transitions
%   \[
%     \begin{array}{r@{}c@{}l}
%       \delta(\sid,\lambda) &~=~&
%       \begin{array}[t]{@{}l}
%         (%%\begin{array}[t]{@{}l@{}}
%           (0,\Id) \land
%           ((0,\lnot\hfc) \lor (1,\nid)) \land
%           ((0,\lnot\hrs) \lor (2,\nid))) \lor{}
%          %%\end{array}
%           \\
%         (%%\begin{array}[t]{@{}l@{}}
%           (0,\lnot\Id) \land
%           (0,\hfc) \land (1,\sid) \land
%           ((0,\lnot\hrs) \lor (2,\nid))) \lor{}
%         %%\end{array}
%           \\
%         (%%\begin{array}[t]{@{}l@{}}
%           (0,\lnot\Id) \land
%           (0,\hrs) \land (2,\sid) \land
%           ((0,\lnot\hfc) \lor (1,\nid))
%         %%\end{array}
%       \end{array}\\
%       \delta(\nid,\lambda) &~=~&
%         %%\begin{array}[t]{@{}l@{}}
%           (0,\lnot\Id) \land
%           ((0,\lnot\hfc) \lor (1,\nid)) \land
%           ((0,\lnot\hrs) \lor (2,\nid))
%         %%\end{array}
%     \end{array}
%   \]
%   Such transitions ensure that exactly one node of the tree is labeled with
%   $\Id$.
\end{enumerate}

\item The set of accepting states is $\alpha=\alpha_\query\cup\{\str\}$.  The
  states $\ini$ and $\str$ form each a single element of the partition of
  states, where $\{\ini\}$ precedes all other elements, and $\{\str\}$ follows
  them.
  % The states $\sid$ and $\nid$ are added to the element of the partition
  % containing all literals.
\end{itemize}
As for the size of $\Atwf_\query$, by Theorem~\ref{thm:muXPath-to-2WATA}, and
considering that the additional states and transitions in $\Atwf_{\query}$ are
of constant size, which does not depend on $\query$, we get that
% \begin{proposition} \label{prop-size-atwf-varphi}
the number of states of $\Atwf_{\query}$ is linear in the size of $\query$.
% \end{proposition}

\begin{theorem}
  \label{thm:satisfiability-to-nonemptiness}
  Let $\query$ be a \muXPath query, and $\Atwf_{\query}$ the corresponding
  2WATA constructed as above.  Then $\Atwf_{\query}$ is nonempty if and only if
  $\query$ is satisfiable.
\end{theorem}

\begin{proof}
\OnlyIf %
Let $\Atwf_{\query}$ accept a binary tree $T$ from a node $x$.  Consider the
first node $r$ encountered on the path from $x$ to the root of $T$ in which
neither $\ifc$ nor $\irs$ hold.  Consider now the subtree $T'$ of $T$ rooted at
$r$, and where every subtree rooted at a node in which neither $\hfc$ nor
$\hrs$ holds is pruned away.  By Transitions~\ref{enu:transition-initial},
and~\ref{enu:transition-structural}
% and~\ref{enu-transition-id}
in the definition of $\Atwf_{\query}$, we have that $T'$ is well-formed, and
hence we can consider the sibling tree $\Ts=\pi_s(T')$ induced by $T'$.  Since
$T'=\pi_b(\Ts)$ and considering that by construction of $\Atwf_{\query}$, $T'$
is accepted by $\At_{\query}$, by Theorem~\ref{thm:muXPath-to-2WATA}, we have
that $\query$ selects $x$ from $\Ts$, and hence is satisfiable.

\If %
If $\query$ is satisfiable, then there exists a sibling tree $\Ts$ and a node
$x$ in $\Ts$ that is selected by $\query$.  By
Theorem~\ref{thm:muXPath-to-2WATA}, $\At_{\query}$ accepts $\pi_b(\Ts)$ from
$x$, and being $\pi_b(\Ts)$ by construction, also $\Atwf_{\query}$ does so.
\end{proof}

From the above result, we obtain a characterization of the computational
complexity for both query satisfiability and query containment.

\begin{theorem}\label{thm:satisfiability-exptime}
  Checking satisfiability of a \muXPath query is \EXPTIME-complete.
\end{theorem}

\begin{proof}
For the upper bound, by Theorem~\ref{thm:satisfiability-to-nonemptiness},
checking satisfiability of a \muXPath query $\query$ can be reduced to checking
nonemptiness of the 2WATA $\Atwf_{\query}$.  $\Atwf_{\query}$ has just two
states more than $\At_{\query}$, which in turn, by
Theorem~\ref{thm:muXPath-to-2WATA} has a number of states that is linear in the
size of $\query$ and an alphabet whose size is exponential in the size of the
alphabet of $\query$.  Finally, by Theorem~\ref{thm-2WATA-nonemptiness}
checking nonemptiness of $\Atwf_{\query}$ can be done in time exponential in
its number of states and linear in the size of its alphabet, from which the
claim follows.

For the hardness, it suffices to observe that satisfiability of \RXPath
queries, which can be encoded in linear time into \muXPath (see
Section~\ref{sec:muxpath}), is already \EXPTIME-hard \cite{ABDG*05}.
\end{proof}

\begin{theorem}\label{thm:containment-exptime}
  Checking containment between two \muXPath queries is \EXPTIME-complete.
\end{theorem}

\begin{proof}
To check query containment $(X_1:\F_1)\subseteq (X_2:\F_2)$, it suffices to
check satisfiability of the \muXPath query
$X_0:\F_1\cup\F_2\cup\{\LFP{X_0=X_1\land\lnot X_2}\}$, where without loss of
generality we have assumed that the variables defined in $\F_1$ and $\F_2$ are
disjoint and different from $X_0$.  Hence, by
Theorem~\ref{thm:satisfiability-exptime}, we get the upper bound.

For the lower bound, it suffices to observe that query $(X_1:\F_1)$ is
unsatisfiable if and only if it is contained in the query
$X_2:\{\LFP{X_0=\false}\}$.
\end{proof}

\subsection{Root Constraints}

Following \cite{Marx04b}, we now introduce \emph{root constraints}, which in
our case are \muXPath formulas intended to be true on the root of the document,
and study the problem of reasoning in the presence of such constraints.
Formally, the root constraint $\varphi$ is \emph{satisfied} in a sibling tree
$\Ts$ if $\eword\in\Int[\Ts]{\varphi}$.
% i.e., the root is in the extension of $\varphi$ in $\Ts$.
%
A (finite) set $\Gamma$ of root constraints is \emph{satisfiable} if there
exists a sibling tree $\Ts$ that satisfies all constraints in $\Gamma$.
A set $\Gamma$ of root constraints \emph{logically implies} a root constraint
$\varphi$, written $\Gamma\models\varphi$, if $\varphi$ is satisfied in every
sibling tree that satisfies all constraints in $\Gamma$.
%
% Note that unsatisfiability and implication of root constraints are mutually
% reducible.  Indeed $\Gamma$ is unsatisfiable if and only if
% $\Gamma\models\FALSE$.  Also, $\Gamma\models\varphi$ if and only if
% $\Gamma\cup\{\lnot\varphi\}$ is unsatisfiable.  Hence, in the following, we
% deal with satisfiability only.

Root constraints are indeed a quite powerful mechanism to describe structural
properties of documents.  For example, as shown in~\cite{Marx05}, \RXPath (and
hence \muXPath) formulas allow one to express all first-order definable sets of
nodes, and this allows for quite sophisticated conditions as root constraints.
In fact, \muXPath differently from \RXPath~\cite{ABDG*05}, can express
arbitrary MSO root constraints (see Section~\ref{sec:NSTA-to-2WATA}).

Also they allow for capturing XML
\emph{DTDs}\footnote{\url{http://www.w3.org/TR/REC-xml/}} by encoding the
right-hand side of DTD element definitions in a suitable path along the \aright
axis.  We illustrate the latter on a simple example (cf.\
also~\cite{CaDL99c,Marx04b} for a similar encoding).

Consider the following DTD element type definition (using grammar-like
notation, with ``,'' for concatenation and ``$|$'' for union), where $A$ is the
element type being defined, and $C$, $D$, $E$ are element types:
\[
  A ~\lora~ B, (C^*|D), E
\]
The constraint on the sequence of children of an $A$-node that is imposed on an
XML document by such an element type definition, can be directly expressed
through the following \RXPath constraint:
\[
  \BOX{\uu}{(
   %%\begin{array}[t]{@{}l}
     A \limp
     \DIAM{\afchild;\TEST{B};
      ((\aright;\TEST{C})^*\cup(\aright;\TEST{D}));
      \aright;\TEST{E}}
     {\BOX{\aright}{\FALSE}})
   %%\end{array}
  }
\]
where $\uu$ is an abbreviation for the path expression
$(\afchild\cup\aright)^*$, and we have assumed to have one atomic proposition
for each element type, and that such proposition are pairwise disjoint (in turn
enforced through a suitable \RXPath constraint).
%
% Observe that $\uu$ connects a node to all nodes that are below it in the
% sibling tree, and $\uu^-;\uu$ connects a node to all nodes in the tree.
% Hence, $\BOX{\uu^-;\uu}\varphi$ holds in some node (or $\BOX{\uu}\varphi$
% holds in the root) if $\varphi$ holds in all nodes in the tree. Similarly
% $\DIAM{\uu^-;\uu}{\varphi}$ (resp., $\DIAM{\uu}{\varphi}$) holds in a node
% $z$, possibly the root, if there is a node in the tree (resp., below $z$)
% where $\varphi$ holds.
%
Similarly, by means of (\RXPath) root constraints, one can express also
\emph{Specialized DTDs}\footnote{Actually, the result that specialized DTDs
 capture MSO sentences dates back to \cite{That67}, which defined what are
 viewed today as unranked tree automata and showed that they capture
 projections (i.e., specialization) of derivation trees of extended context
 free grammars (i.e., DTDs).}~\cite{PaVi00} and the structural part of XML
Schema Definitions\footnote{\url{http://www.w3.org/TR/xmlschema-1}}
(cf.~\cite{Marx04b}).

% Also, they do not allow for capturing (regular) path
% constraints~\cite{AbVi99,GrTh03}, i.e., constraints expressing that all nodes
% reachable through a certain specified path are also reachable through another
% specified path. Finally, extension of \RXPath root constraints with loop
% operators have been considered in the literature~\cite{GoMa05}.

\XPath includes identifiers, which are special propositions that hold in a
single node of the sibling tree.  It is easy to see that the following root
constraint $\Nominal{A}$ forces a proposition $A$ to be an identifier:
\[
  \Nominal{A} ~=~  \begin{array}[t]{ll}
    \DIAM{\uu}{A} \land{}  & (1)\\
    \BOX{\uu} (
    \begin{array}[t]{@{}l}
      (\DIAM{\afchild;\uu}{A}\limp
       \BOX{\aright;\uu}{\lnot A}) \land{}\\
      (\DIAM{\aright;\uu}{A}\limp
       \BOX{\afchild;\uu}{\lnot A}) \land{}\\
      (A \limp
       \BOX{(\afchild\cup\aright);\uu}{\lnot A}))
    \end{array} &
    \begin{array}[t]{@{}l}
      (2)\\ (3)\\ (4)
    \end{array}
  \end{array}
\]
In the above constraint, Line~1 expresses that there exists a node of the tree
where $A$ holds.
Line~2 expresses that, if a node where $A$ holds exists in the
$\afchild$ subtree of a node $n$, then $A$ never holds in the
$\aright$ subtree of $n$.
Line~3 is analogous to Line~2, with $\afchild$ and $\aright$ swapped.
Finally, Line~4 expresses that, if $A$ holds in a node $n$, then
it holds neither in the $\afchild$ nor in the $\aright$ subtree of
$n$.

It is immediate to see that every set $\{X_1:\F_1, \ldots X_k:\F_k\}$ of
\muXPath root constraint can be expressed as a \muXPath query that selects only
the root of the tree:
\[
  X_r:\{\LFP{X_r\doteq (\BOX{\afchild^-}{\FALSE}) \land
   (\BOX{\aright^-}{\FALSE}) \land X_1\land\cdots\land X_k}\} \cup
  \F_1\cup\cdots\cup\F_k.
\]
As a consequence, checking for satisfiability and logical implication of root
constraints can be directly reduced to satisfiability and containment of
\muXPath queries.  Considering that satisfiability and logical implication is
already \EXPTIME-hard for \RXPath root constraints \cite{Marx04b}, we get the
following result.

\begin{theorem}
  \label{thm:root-constraints-complexity}
  Satisfiability and logical implication of \muXPath root constraints are
  \EXPTIME-complete.
\end{theorem}

We can also consider query satisfiability and query containment under root
constraints, i.e., with respect to all sibling trees that satisfy the
constraints.  Indeed, a \muXPath query $X_q:\F_q$ can be expressed as the root
constraint:
\[
  X_r:\{\LFP{X_r\doteq X_q\lor (\DIAM{\afchild}{X_r}) \lor
   (\DIAM{\aright}{X_r})}\} \cup \F_q.
\]
Hence, we immediately get the following result.

\begin{theorem}
  \label{thm:query-root-constraints-complexity}
  Satisfiability and containment of \muXPath queries under \muXPath root
  constraints are \EXPTIME-complete.
\end{theorem}

\begin{proof}
The upper bound follows from Theorem~\ref{thm:root-constraints-complexity}.
The lower bound follows from Theorems~\ref{thm:satisfiability-exptime}
and~\ref{thm:containment-exptime}, by considering an empty set of constraints.
\end{proof}

\subsection{View-based Query Processing}

View-based query processing is another form of reasoning that has recently
drawn a great deal of attention in the database community~\cite{Hale00,Hale01}.
In several contexts, such as data integration, query optimization, query
answering with incomplete information, and data warehousing, the problem arises
of processing queries posed over the schema of a virtual database, based on a
set of materialized views, rather than on the raw data in the
database~\cite{AbDu98,Lenz02,Ullm97}.  For example, an information integration
system exports a global virtual schema over which user queries are posed, and
such queries are answered based on the data stored in a collection of data
sources, whose content in turn is described in terms of views over the global
schema.  In such a setting, each data source corresponds to a materialized
view, and the global schema exported to the user corresponds to the schema of
the virtual database.  Notice that typically, in data integration, the data in
the sources are correct (i.e., sound) but incomplete with respect to their
specification in terms of the global schema.  This is due the fact that
typically the global schema is not designed taking the sources into account,
but rather the information needs of users.  Hence it may not be possible to
precisely describe the information content of the sources.  In this paper we
will concentrate on this case (\emph{sound views}), cf.~\cite{Lenz02}.

Consider now a sibling tree that is accessible only through a collection of
views expressed as \muXPath queries, and suppose we need to answer a further
\muXPath query over the tree only on the basis of our knowledge on the views.
Specifically, the collection of views is represented by a finite set $\V$ of
\emph{view symbols}, each denoting a set of tree nodes.  Each view symbol
$V\in\V$ has an associated \emph{view definition} $\query_V$ and a \emph{view
 extension} $\E_V$.
The view definition $\query_V$ is simply a \muXPath query.  The view extension
$\E_V$ is a set of \emph{node references}, where each node reference is either
an identifier, or an explicit path expression that is formed only by chaining
$\afchild$ and $\aright$ and that identifies the node by specifying how to
reach it from the root.  Observe that a node reference $a$ is interpreted in a
sibling tree $\Ts$ as a singleton set of nodes $\Int[\Ts]{a}$.  We use
$\INT[\Ts]{\E_V}$ to denote the set of nodes resulting from interpreting the
node references in $\Ts$.
We say that a \emph{sibling tree $\Ts$ satisfies a view $V$} if
$\INT[\Ts]{\E_V}\subseteq\INT[\Ts]{\query_V}$.  In other words, in $\Ts$ all
the nodes denoted by $\INT[\Ts]{\E_V}$ must appear in $\INT[\Ts]{\query_V}$, but
$\INT[\Ts]{\query_V}$ may contain nodes not in $\INT[\Ts]{\E_V}$.
% \footnote{We consider views to be \emph{sound}~\cite{AbDu98,GrMe99}, i.e., we
% model a situation where the extension of the views provides a subset of the
% results of applying the view definitions to the document (the sibling tree).}

Given a set $\V$ of views, and a \muXPath query $\query$, the set of
\emph{certain answers} to $\query$ with respect to $\V$ under root constraints
$\Gamma$ is the set $\cert[\query,\V,\Gamma]$ of node references $a$ such that
$\Int[\Ts]{a}\in \Int[\Ts]{\query}$, for every sibling tree $\Ts$ satisfying
each $V\in\V$ and each constraint in $\Gamma$.  \emph{View-based query
 answering} under root constraints consists in deciding whether a given node
reference is a certain answer to $\query$ with respect to $\V$.

View-based query answering can also be reduced to satisfiability of root
constraints.
Given a view $V$, with extension $\E_V$ and definition $X_V:\F_V$, for each
$a\in \E_V$:
\begin{itemize}
\item if $a$ is an identifier, then we introduce the root
  constraint
  \[
    X_a:\{\LFP{X_a\doteq a\land X_V\lor (\DIAM{\afchild}{X_a}) \lor
     (\DIAM{\aright}{X_a})}\} \cup \F_V.
  \]
\item if $a$ is an explicit path expressions $P_1;\cdots;P_n$, then we
  introduce the root constraint
  \[
    X_a:\{\LFP{X_a\doteq \DIAM{P_1}{\cdots\DIAM{P_n}{X_V}}}\}\cup \F_V.
  \]
\end{itemize}
Let $\Gamma_\V$ be the set of \muXPath root constraints corresponding to the
set of \muXPath views $\V$, $\query=X_q:\F_q$ a \muXPath query, and $\Gamma$ a
finite set of root constraints.  Then a node reference $c$ belongs to
$\cert[\query,\V,\Gamma]$ if and only if the following set of root constraints
is unsatisfiable:
\[
  \Gamma \cup \Gamma_\V \cup \Gamma_{id} \cup \Gamma_{\lnot\query},
\]
where $\Gamma_{id}$ consists of one root constraint $N_a$ imposing that $a$
behaves as an identifier (see above), for each node $a$ appearing in $\V$, and:
\begin{itemize}
\item if $c$ is and identifier, then $\Gamma_{\lnot\query}$ is
  \[
    X_c:\{\LFP{X_c\doteq c\land \lnot X_q\lor (\DIAM{\afchild}{X_c}) \lor
     (\DIAM{\aright}{X_c})}\} \cup \F_q.
  \]
\item if $c$ is an explicit path expressions $P_1;\cdots;P_n$, then
  $\Gamma_{\lnot\query}$ is
  \[
    X_c:\{\LFP{X_c\doteq \DIAM{P_1}{\cdots\DIAM{P_n}{\lnot X_q}}}\}\cup \F_q.
  \]
\end{itemize}
Hence we have linearly reduced view-based query answering under root
constraints to unsatisfiability of \muXPath root constraints, and the following
result immediately follows.

\begin{proposition}\label{view-based-query-answering}
  View-based query answering under root constraints in \muXPath is
  \EXPTIME-complete.
\end{proposition}

%%% Local Variables:
%%% mode: latex
%%% TeX-master: "main"
%%% End:

\section{Relationship among \muXPath, 2WATAs, and MSO}
\label{sec:2wata-mso}

% \todo{Moshe}

In this section, we show that \muXPath is expressively equivalent to Monadic
Second-Order Logic (MSO)\footnote{Our result is analogous to that by
 \cite{Mate02}, who showed that on acyclic labeled transition systems, the full
 $\mu$-Calculus collapses to the Alternation-Free $\mu$-Calculus, though the
 construction there leads to a higher query complexity.}.  We have already
shown that \muXPath is equivalent to 2WATAs, hence it suffices to establish the
relationship between 2WATAs and MSO.  To do so we make use of nondeterministic
node-selecting tree automata, which were introduced in \cite{FrGK03}, following
earlier work on deterministic node-selecting tree automata in \cite{NeSc02}.
(For earlier work on MSO and Datalog, see \cite{GoKo02,GoKo04}.)  For technical
convenience, we use here top-down, rather than bottom-up automata. It is also
convenient here to assume that the top-down tree automata run on \emph{full}
binary trees, even though our binary trees are not full. Thus, we can assume
that there is a special label $\bot$ such that a node that should not be
present in the tree (e.g, left child of a node that does not contain $\hfc$ in
its label) is labeled by $\bot$.

A \emph{nondeterministic node-selecting top-down tree
automaton} (NSTA) on binary trees is a tuple
$\At=(\L,S,S_0,\delta,F,\sigma)$, where $\L$ is the alphabet of tree labels,
$S$ is a finite set of states, $S_0\subseteq S$ is the initial state set,
$\delta:S\times\L\ra 2^{S^2}$ is the transition function,
$F\subseteq S$ is a set of accepting states, and
$\sigma\subseteq S$ is a set of selecting states. Given a tree
$T=\Tree[T]$, an \emph{accepting run} of $\At$ on $T$ is an $S$-labeled tree
$R=(\Delta^T,\ell^R)$, with the same node set as $T$, where:
\begin{itemize}
\item $\ell^R(\eword)\in S_0$.
\item If $x\in\Delta^T$ is an interior node, then $\langle \ell^R(x\cdot
  1),\ell^R(x\cdot 2)\rangle \in\delta(\ell^R(x),\ell^T(x))$.
\item If $x\in\Delta^T$ is a leaf, then $\delta(\ell^R(x),\ell^T(x)) \cap F^2
  \neq\emptyset$.
\end{itemize}
A node $x\in\Delta^T$ is \emph{selected} by $\At$ from $T$ if there is a
run $R=(\Delta^T,\ell^R)$ of $\At$ on $T$ such that $\ell^R(x)\in\sigma$.
The notion of accepting run used here is standard, cf.\ \cite{CDGJ*02}.
It is the addition of selecting states that turns these trees from a
model of tree recognition to a model of tree querying.

\begin{theorem} {\rm \cite{FrGK03}}
\label{thm:nsta-mso}
%\begin{inparaenum}[(i)]
%\item
\textit{(i)}~For each MSO query $\varphi(x)$, there is an NSTA $\At_\varphi$
such that a node $x$ in a tree $T=(\Delta^T,\ell^T)$ satisfies $\varphi(x)$ iff
$x$ is selected from $T$ by $\At_\varphi$.
%\item
\textit{(ii)}~For each NSTA $\At$, there is an MSO query $\varphi_{\At}$ such
that a node $x$ in a tree $T=(\Delta^T,\ell^T)$ satisfies $\varphi_{\At}(x)$
iff $x$ is selected from $T$ by $\At$.
%\end{inparaenum}
\end{theorem}

We now establish back and forth translations between 2WATAs and NSTAs, implying
the equivalence of 2WATAs to MSO.

\subsection{From 2WATAs to NSTAs}
\label{2WATA-to-NSTA}

\begin{theorem}
\label{thm:2wata-to-nsta}
For each 2WATA $\At$, there is an NSTA $\At'$ such that a node $x$ in a
binary tree $T$ is selected by $\At$ if and only if it is selected by $\At'$.
\end{theorem}

\begin{proof}
In Section~\ref{sec:2wata}, we described a translation of 2WATA to NTA.  Both
the 2WATA and the NTA start their runs there from the root $\varepsilon$ of the
tree. Here we need the 2WATA to start its run from a node $x_0\in\Delta^T$, on
one hand, and we want the NSTA to select this node $x_0$. Note, however, that
the fact that the 2WATA starts its run from $\varepsilon$ played a very small
in the construction in Section~\ref{sec:2wata}. Namely, we defined the set of
initial states as:
$Q_0=\{(r,R): s_0\in\sta(r) \text{ and there is no transition }
(s,-1,s')\in r\}$.  The requirement that $s_0\in\sta(r)$ corresponds
to the 2WATA starting its run in $\varepsilon$. This captured the requirement
that we have an \emph{anchored} accepting strategy for $T$.

More generally, however, we can say that the strategy $\tau$ is \emph{anchored
 at a node} $x_0\in\Delta^T$ if we have $s_0\in\sta(\tau(x_0))$.

Then we can relax the claims in Section~\ref{sec:2wata}:
\begin{claim}{\rm \cite{Vard98}} \label{tree3}
\begin{enumerate}
\item
A node $x_0$ of $T$ is selected by the 2WATA $\At$
iff $\At$ has an accepting strategy for $T$ that is
anchored at $x_0$.
\item
A node $x_0$ of $T$ is selected by the 2WATA $\At$
iff $\At$ has an strategy for $T$ that is anchored at $x_0$
and an accepting annotation $\eta$ of $\tau$.
\end{enumerate}
\end{claim}

To match this relaxation in the construction of the NSTA, we need to redefine
the set of initial states as $Q_0=\{(r,R): \mbox{there is no transition} (s,-1,s')\in r\}$,
which means that  the requirement that the 2WATA starts its run from the root is dropped,
as the strategy guessed by the NSTA no longer need to be anchored at $\varepsilon$.
Instead, we want the strategy to be anchored at the node $x_0$ selected by $\At$.
To that end, we define the set of selecting states as
$\sigma=\{(r,R)\in \Tau\times\Eta: s_0\in\sta(r)\}$. That is, if $\At$ starts
its run at $x_0$, then the strategy needs to be anchored at $x_0$ and $x_0$ is selected
by the NSTA.

Finally, that the NTA constructed in Section~\ref{sec:2wata} accepts
when the transition function yields the truth value {\bf true}, while NSTA
accepts by means of accepting states. We can simply add a special accepting state
$accept$ and transition to it whenever the transition function yields {\bf true}.
\end{proof}

Finally, we remark that the translation from 2WATA to NSTA is exponential.
Together with the results in the previous sections, we get an exponential
translation from \muXPath to NSTAs.  This explains why NSTAs are not useful for
efficient query-evaluation algorithms, as noted in \cite{Schw07}.

\subsection{From NSTAs to 2WATAs}
\label{sec:NSTA-to-2WATA}

For the translation from NSTAs to 2WATAs, the idea is to take an accepting run
of an NSTA, which starts from the root of the tree, and convert it to a run of
a 2WATA, which starts from a selected node.  The technique is related to the
translation from tree automata to Datalog in \cite{GoKo04}. The construction
here uses the propositions $\ifc$, $\irs$, $\hfc$, and $\hrs$ introduced
earlier.

\begin{theorem}
\label{thm:nsta-to-2wata}
For each NSTA $\At$, there is a 2WATA $\At'$ such that a node $x_0$ in a tree
$T$ is selected by $\At$ if and only if it is selected by $\At'$.
\end{theorem}
\begin{proof}
Let $\At=(\L,S,S_0,\delta,F,\sigma)$ be an NSTA. We construct an equivalent
2WATA $\At'=(\L,S',s'_0,\delta',\alpha')$ as follows (where $s\in S$ and
$a\in\L$):
\begin{itemize}
\item
$S'=\{s_0\}\cup S\times\{u,d,l,r\}\cup \Sigma$. (We add a new initial
state, and we keep four copies, tagged with $u$, $d$, $l$, or $r$ of each
state in $S$. We also add the alphabet to the set of states.)
\item
$\alpha'=\emptyset$. (Infinite branches are not allowed in runs of $\At'$.)
\item
$\delta'(s_0,a)=\bigvee_{s\in \sigma} (((s,d),0) \wedge ((s,u),0))$.
($\At'$ guesses an accepting state of $\At$ and spawns two copies, tagged with
$d$ and $u$, to go downwards and upwards.)
\item
If $a$ does not contain $\hfc$ and does not contain $\hrs$ (that is, we
are reading a leaf node), then $\delta'((s,d),a)=\true$ if $\delta(s,a)\cap
F^2 \neq\emptyset$, and $\delta'((s,d),a)=\false$ if $\delta(s,a)\cap F^2
=\emptyset$.  (In a leaf node, a transition from $(s,d)$ either accepts or
rejects, just like $\At$ from $s$.)
\item
If $a$ contains $\hfc$ or $\hrs$ (that is, we are reading an interior
node), then $\delta'((s,d),a)=\bigvee_{(t_1,t_2)\in\delta(s,a)} ((t_1,d),1)
\wedge ((t_2,d),2)$. (States tagged with $d$ behave just like the
corresponding states of $\At$.)
\item
If $a$ contains neither $\ifc$ nor $\irs$ (that is, we
are reading the root node), then $\delta'((s,u),a)=\true$
if $s\in S_0$, and $\delta'((s,u),a)=\false$, otherwise
(that is, if an upword state reached the root, then we just
need to check that the root has been reached with an initial state),
\item
If $a$ contains $\ifc$ (it is a left child), then
$\delta'((s,u),a)=\bigvee_{t\in S, a'\in\L, (s,t')\in\delta(t,a')}
((t,u),-1) \wedge (a',-1) \wedge ((t',r),-1)$.
(Guess a state and letter in the node above, and proceed to check them.)
\item
If $a$ contains $\irs$ (it is a right child), then $\delta'((s,u),a)=
\bigvee_{t\in S, a'\in\L, (t',s)\in\delta(t,a')}
((t,u),-1) \wedge (a',-1) \wedge ((t',l),-1)$.
(Guess a state and letter in the node above, and proceed to check them.)
\item
$\delta'(a',a)=\true$ if $a'=a$ and $\delta'(a',a)=\false$ if $a'\neq a$.
(Check that the guessed letter was correct.)
\item
$\delta'((s,l),a)=((s,d),1)$. (Check left subtree.)
\item
$\delta'((s,r),a)=((s,d),2)$. (Check right subtree.)
\qed
\end{itemize}
Intuitively, $\At'$ tries to guess an accepting run of $\At$ that selects
$x_0$. It starts at $x_0$ in an accepting state, guesses the subrun below
$x_0$, and also goes up the tree to guess the rest of the run. Note that
we need not worry about cycles in the run of $\At'$, as it only goes upward
in the $u$ mode, and once it leaves the $u$ mode it never enters again the $u$
mode.

We need to show that that a node $x_0$ in a tree $T$ is selected by $\At$
if and only if it is selected by $\At'$. If a node $x_0$ of $T$ is selected
by $\At$, then $\At$ has an accepting run on $T$ that reaches $x_0$ in some
accepting state $s_a\in\sigma$.  Then $\At'$ starts its run at $x_0$ in state
$s_a$, and it proceeds to emulate precisely the accepting run of $\At$. More
precisely, at $x_0$ $\At'$ branches conjunctively to both $(s_a,d)$ and $(s_a,u)$.
From $(s_a,d)$, $At'$ continues downwards and emulate the run of $\At$.
That is, if $\At$ reaches a node $x$ below $t$ in state $s$, then $\At'$
reaches $x$ in state $(s,d)$. At the leaves, $\At$ transitions to accepting
states, and $\At'$ transitions to $\true$. From $(s_a,u)$, $\At'$ first
continues upwards, If a $x$ above $x_0$ is labeled by the letter $a$
and is reached by $\At$ in state $t$, it is reached by $\At'$ from its left
child, and $\At$ reaches its right child in state $t$'.
then $\At'$ reaches $x$ with $((t,u)$, $a$, and $(t',r)$. Then $\At'$
continues the upward emulation from $(t,u)$, verify that $a$ is the letter at $x$,
and also transition to the right child in state $(t',d)$, from which it
continues with the downward emulation of $\At$. The upward emulation
eventually reaches the root in an initiai state.

On the other hand, for $\At'$ to select $x_0$ in $T$ it must start at $x_0$
in some state $s_a\in\sigma$. While $\At'$ is a 2WATA and its run is a
run tree, this run tree has a very specific structure. From a state $(s,d)$,
$\At'$ behaves just like an NTA. From a state $(s,u)$, $\At'$ proceeds upward,
trying to label every node on the path to the root with a single state,
such that the root is labeled by an initial state, and from every node on
that path $At'$ can then go downward, again labeling each node by a single
state. Thus, $At'$ essentially guess an accepting run tree of $\At$ that
selects $x_0$.
\end{proof}

While the translation from 2WATAs to NSTAs was exponential, the translation
from NSTAs to 2WATAs is linear.  It follows from the proof of
Theorem~\ref{thm:nsta-to-2wata} that the automaton $\At'$ correspond to
$\lfp$-\muXPath, which consists of \muXPath queries with a single, least
fixpoint block.  This clarifies the relationship between \muXPath and
Datalog-based languages studied in \cite{GoKo04,FrGK03}. In essence,
\muXPath corresponds to stratified monadic Datalog, where rather than
use explicit negation, we use alternation of least and greatest fixpoints,
while $\lfp$-\muXPath corresponds to monadic Datalog. The results of the last
two sections provide an exponential translation from \muXPath to
$\lfp$-\muXPath. Note, however, that $\lfp$-\muXPath does not have a
computational advantage over \muXPath, for either query evaluation or
query containment. In contrast, while stratified Datalog queries can be
evaluated in polynomial time (in terms of data complexity), there is no good
theory for containment of stratified monadic Datalog queries.

\medskip

The above results provide us a characterization of the expressive power of
\muXPath.

\begin{theorem}\label{thm:muXPath-mso}
  Over (binary) sibling trees, \muXPath and MSO have the same expressive power.
\end{theorem}

\begin{proof}
By Theorems~\ref{thm:muXPath-to-2WATA} and~\ref{thm:2WATA-to-muXPath}, \muXPath
is equivalent to WATAs, and by Theorems~\ref{thm:nsta-mso},
\ref{thm:2wata-to-nsta}, and~\ref{thm:nsta-to-2wata}, 2WATAs are equivalent to
MSO.
\end{proof}

%%% Local Variables:
%%% mode: latex
%%% TeX-master: "main"
%%% End:

\section{Conclusion}
\label{sec:conclusion}

%% \todo{Maurizio}

The results of this paper fill a gap in the theory of node-selection queries
for trees.  With a natural extension of \XPath by fixpoint operators, we
obtained \muXPath, which is expressively equivalent to MSO, has linear-time
query evaluation and exponential-time query containment, as \RXPath.  2WATAs,
the automata-theoretic counterpart of \muXPath, fills another gap in the theory
by providing an automaton model that can be used for both query evaluation and
containment testing.  Unlike much of the theory of automata on infinite trees,
which so far has resisted implementation, the automata-theoretic machinery over
finite trees should be much more amenable to practical implementations.
%MYV:
Pursuing this, is a promising future research direction.

Our automata-theoretic approach is based on techniques developed in the context
of program logics~\cite{KuVW00,Vard98}. Here, however, we leverage the fact
that we are dealing with finite trees, rather than the infinite trees used in
the program-logics context. Indeed, the automata-theoretic techniques used in
reasoning about infinite trees are notoriously difficult and have resisted
efficient
%MYV: Added a reference.
implementation~~\cite{FrLa10,ScTW05,TaHB95}. The restiction to finite trees here
enables us to obtain much more feasible algorithmic approach.  In particular,
as pointed out in \cite{CDLV09}, one can make use of symbolic techniques, at
the base of modern model checking tools, for effectively querying and verifying
XML/JSON documents.  It is worth noting that while our automata run over finite
trees they are allowed to have infinite runs.  This separates 2WATAs from the
alternating tree automata used in~\cite{CGKV88,Slut85}.  The key technical
results here are that acceptance of trees by 2WATAs can be decided in linear
time, while nonemptiness of 2WATAs can be decided in exponential time.

%MYV: Added
While our focus here has been on querying finite unranked trees, recent work
has considered using languages for querying graph databases such as \XPath
\cite{LiMV16,FrLi17} and variants of regular path queries
\cite{Barc13,Vard16,ReRV17}. It'd be interesting to follow this line of
research of querying graph databases with \muXPath.

%%% Local Variables:
%%% mode: latex
%%% TeX-master: "main"
%%% End:

% \section*{Acknowledgment}
% \noindent
% This research has been partially supported by \dots
% NSF grants CCR-0124077, CCR-0311326, CCF-0613889, ANI-0216467, and
% CCF-0728882.

\bibliographystyle{alpha}
%\bibliography{medium-string,krdb,w3c}
%\bibliography{main-bib}
\newcommand{\etalchar}[1]{$^{#1}$}

\end{document}

%%% Local Variables:
%%% mode: latex
%%% TeX-master: t
%%% save-place: t
%%% End: